\documentclass{article}
\usepackage[utf8]{inputenc}
\pdfoutput=1
\usepackage[sort,comma,square,numbers]{natbib}
\usepackage[bookmarks=false, colorlinks,
 linkcolor=red!60!black, citecolor=green!60!black]{hyperref}
\usepackage[protrusion=true,expansion=true]{microtype}
\usepackage{xcolor}
\usepackage{tabulary}
\usepackage{tabularx}
\usepackage{multirow}
\usepackage{amsthm}
\usepackage{amsmath}
\usepackage{microtype}
\usepackage{mathtools}
\usepackage{booktabs}
\usepackage{paralist}
\usepackage{amssymb}
\usepackage[textsize=scriptsize,linecolor=black,backgroundcolor=white]{todonotes}
\usepackage{mdwlist}
\usepackage{authblk}

\makeatletter
\def\NAT@spacechar{~}%
\makeatother

\newenvironment{mathenum}{%
  \begin{enumerate}[(i)]%
  }{%
  \end{enumerate}%
}
\newcommand{\enuref}[1]{(\ref{#1})}

\newcommand{\prob}[5]{%
  \begin{center}
    \begin{quote}
      #1\nopagebreak
      \begin{compactdesc}
      \item[#2]#3
      \item[#4]#5
      \end{compactdesc}
    \end{quote}
  \end{center}
}

\newcommand{\decprob}[3]{\prob{#1}{Input:}{#2}{Question:}{#3}}

\newcommand{\wf}{\ensuremath{\lambda}}
\newcommand{\WF}{\ensuremath{\Lambda}}
\DeclareMathOperator{\suc}{succ}
\DeclareMathOperator{\poly}{poly}

\newcommand{\vcn}{\ensuremath{\tau}}

\newcommand{\tw}{\ensuremath{\omega}}

\newcommand{\cvsep}[2]
{\ensuremath{#1}\nobreakdash-\ensuremath{#2}\nobreakdash-separator}
\newcommand{\pth}[2]
{\ensuremath{#1}\nobreakdash-\ensuremath{#2}\nobreakdash-path}

\newcommand{\BB}{\textsc{Bal\-anced Par\-ti\-tion\-ing}}
\newcommand{\ClV}{\textsc{Cut\-ting $\ell$ Ver\-tices}}
\newcommand{\EWBS}{\textsc{Edge-Weighted Bi\-sec\-tion}}
\newcommand{\MaxCu}{\textsc{Max\-i\-mum Cut}}

\newcommand{\BP}{\textsc{Bi\-sec\-tion}}
\newcommand{\VBP}{\textsc{Ver\-tex Bi\-sec\-tion}}
\newcommand{\UBP}{\textsc{Unary Bin Packing}}

\theoremstyle{theorem}
\newtheorem{theorem}{Theorem}

\theoremstyle{definition}
\newtheorem{construction}{Construction}%
\newtheorem{definition}{Definition}
\newtheorem{corollary}{Corollary}
\newtheorem{lemma}{Lemma}
\newtheorem{remark}{Remark}
\newtheorem{proposition}{Proposition}

\date{}

\title{On the Parameterized Complexity of Computing Balanced
  Partitions in Graphs\thanks{An extended abstract of this article
    appeared at the \emph{39th International Workshop on
      Graph-Theoretic Concepts in Computer Science
      (WG~2013)}~\cite{BFSS13}. The extended abstract contains results
    regarding \BP{} and \VBP{}, while this article additionally
    provides full proof details as well as parameterized complexity
    analyses of \BB{}. This article is published in \emph{Theory of
      Computing Systems}, available at
    \href{http://dx.doi.org/10.1007/s00224-014-9557-5}{link.springer.com}.}}

\author[1]{René van Bevern}
\affil[1]{Institut für Softwaretechnik und Theoretische Informatik, TU Berlin,
  Germany, \texttt{\{rene.vanbevern,manuel.sorge\}@tu-berlin.de}}
\author{Andreas Emil Feldmann}
\affil{Combinatorics \& Optimization, University of Waterloo, Canada, \texttt{andreas.feldmann@uwaterloo.ca}}
\author[1]{Manuel Sorge}
\author{Ondřej Such\'y}
\affil{Faculty of Information Technology, Czech Technical University in 
 Prague, Czech Republic, \texttt{ondrej.suchy@fit.cvut.cz}}

\begin{document}

\maketitle

\begin{abstract}
\looseness=-1
A balanced partition is a clustering of a graph into a given number of 
equal-sized parts. For instance, the \BP{} problem asks to remove at most $k$ 
edges in order to partition the vertices into two equal-sized parts. 
We prove that \BP{} is FPT for the distance to constant cliquewidth if we are 
given the deletion set. This implies FPT algorithms for some well-studied 
parameters such as cluster vertex deletion number and feedback vertex set. 
However, we show that \BP{} does not admit polynomial-size kernels for these 
parameters. 

For the \VBP{} problem, vertices need to be removed in order to obtain two 
equal-sized parts. We show that this problem is FPT for the number of 
removed vertices $k$ if the solution cuts the graph into a constant number $c$
of connected components. The latter condition is unavoidable, since we also 
prove that \VBP{} is W[1]-hard w.r.t.~$(k,c)$. 

\looseness=-1
Our algorithms for finding bisections can easily be adapted to finding 
partitions into $d$~equal-sized parts, which entails additional running 
time factors of~$n^{O(d)}$. We show that a substantial speed-up is unlikely 
since the corresponding task is W[1]-hard w.r.t.~$d$, even on forests of 
maximum degree two. We can, however, show that it is FPT for the vertex cover 
number. 
\end{abstract}

\section{Introduction}
\label{sec:intro}
In this article we consider partitioning problems on graphs. These are 
clustering-type problems in which the clusters need to be equal-sized, and the 
number of edges connecting the clusters needs to be minimized. At the same time 
the desired number of clusters is given. We begin with the setting in which only 
two clusters need to be found, and later generalize to more clusters. In 
particular we consider the \BP{}, \VBP{}, and \BB{} problems, which are defined 
below. We study these problems from a parameterized complexity point of view 
and consider several parameters that naturally arise from the known results (see 
\autoref{table:results}). 
That is, we consider a given parameter $p$ of an input instance and 
ask whether an algorithm with running time $f(p)\cdot n^{O(1)}$ exists that 
optimally solves the problem. Here $n$~is the instance size and $f(p)$~is a 
function that only depends on~$p$. If there is such an algorithm, then the 
problem is called \emph{fixed-parameter tractable} (or FPT for short) with 
respect to~$p$. For in-depth introductions to parameterized complexity we refer 
to the literature~\cite{Nie06, FG06, DF13}. Throughout this article we use 
standard terminology of graph theory~\cite{Diestel10}. 

\begin{table}
\centering\footnotesize
  \caption{Overview of known and new parameterized results.}
  \begin{tabular}{p{.13\textwidth} p{.37\textwidth} p{.37\textwidth}} 
   \toprule
   Problem & Parameter & Results \\%

   \midrule
   \BP{}&
    cut size  &
	FPT for planar graphs~\cite{BuiP92} \newline
	FPT in general~\cite{CLPPS14} \newline
	 No poly-size kernel (\autoref{the:no-poly-kernel-tw-k})\\
    \cmidrule{2-3}
    &treewidth  &
	 FPT~\cite{SoumyanathD90,Wie90}\newline
	 No poly-size kernel (\autoref{the:no-poly-kernel-tw-k})\\
    \cmidrule{2-3}
    &union-oblivious (e.g.\ bandwidth) &
	 No poly-size kernel (\autoref{the:no-poly-kernel-tw-k}) \\
    \cmidrule{2-3}
    &cliquewidth-$q$ deletion number &
	 FPT (\autoref{cor:constant-cliquewidth-deletion-set}) \\
    \cmidrule{2-3}
   & cliquewidth  &
         XP, W[1]-hard~\cite{FGLS10} \\

   \midrule
   \multirow{2}{.16\textwidth}{\VBP{}}&
     cut size&
	 FPT if nr of cut out components
         is constant (\autoref{thm2})\\
   \cmidrule{2-3}
   &
   cut size \& \newline nr of cut out components&
	 W[1]-hard (\autoref{vbswhard})\\

    \midrule
    \multirow{2}{.16\textwidth}{\BB{}}&
    cut size \& \newline nr of cut out components
 	& W[1]-hard (\autoref{prop:bb-whard-d-treewidth-k})\\
    \cmidrule{2-3}
    &treewidth  
	& NP-hard for trees~\cite{AFFoschini12}\\
    \cmidrule{2-3}
    &cliquewidth
	& NP-hard for cluster graphs~\cite{RackeA06}\\
    \cmidrule{2-3}
    &vertex cover
        & FPT (\autoref{bb-vc-fpt}) \\
    \cmidrule{2-3}
    &nr $d$ of parts in the partition
	& W[1]-hard for forests (\autoref{thm:bb-whard-forests})\\
    \bottomrule
  \end{tabular}
  \label{table:results}
\end{table}

\subsection{The \BP{} problem}
The first problem we consider is the 
NP-hard~\cite{GareyJS76} \BP{} problem for which the $n$~vertices of a 
graph~$G=(V,E)$ need to be partitioned into two parts $A$ and $B$ of size at 
most~$\lceil n/2\rceil$ each, while minimizing the number of edges connecting 
$A$ and $B$. The partition~$\{A,B\}$ is called a \emph{bisection} of~$G$, and  
the number of edges connecting vertices in $A$ with vertices in~$B$ is called 
the \emph{cut size}. Throughout this article it will be convenient to consider 
\BP{} as a decision problem, which is defined as follows. 

\decprob{\BP{}}
{A graph~$G$ and a positive integer~$k$.}
{Does~$G$ have a bisection with cut size at most~$k$?} 

The \BP{} problem is of importance both in theory and practice, and for instance 
has applications in divide-and-conquer algorithms~\cite{LiptonT80}, computer 
vision~\cite{KwatraSETB03}, and route planning~\cite{DellingGPW11}. 
As a consequence, the problem has been thoroughly studied in the past. It is 
known that it is NP-hard in general~\cite{GareyJS76} and that the minimum cut size 
can be approximated within a factor of~$O(\log n)$~\cite{Racke08}. Assuming the 
Unique Games Conjecture, no constant factor approximations exist~\cite{KhotV05}. 
For special graph classes such as trees~\cite{MacG78} and solid 
grids~\cite{AFWidmayer11} the optimum cut size can be computed in polynomial 
time. For planar graphs it is still open whether \BP{} is NP-hard, but it is 
known to be FPT with respect to the cut size~\cite{BuiP92}.

It was recently shown by~\citet{CLPPS14} that \BP{} is FPT with respect to 
the cut size on general graphs.
We complement this result by showing that \BP{} does not allow for
polynomial-size problem kernels for this parameter unless
coNP${}\subseteq{}$NP/poly.  Hence, presumably there is no
polynomial-time algorithm that reduces an instance of \BP{} to an
equivalent one that has size polynomial in the desired cut size. We
prove this by giving a corresponding result for all parameters that
are polynomial in the input size and that do not increase when taking
the disjoint union of graphs. We call such parameters
\emph{union-oblivious}. This includes parameters such as treewidth,
cliquewidth, bandwidth, and others.

Some of these parameters have been considered for the \BP{} problem before. 
For instance, we already mentioned the cut size, and it was shown that the
problem is FPT with respect to treewidth~\cite{SoumyanathD90,Wie90}. Even 
though treewidth is probably the most widely used graph parameter for sparse 
graphs, it is not suitable for dense graphs, although they can also have simple 
structure. For that purpose, \citet{CO00} introduced the parameter 
\emph{cliquewidth}~\cite{EGW01}. \citet{FGLS10} showed that \BP{} is W[1]-hard with respect to cliquewidth, that is, an FPT-algorithm is unlikely. On the positive side, they give an $n^{O(q)}$-time algorithm if a cliquewidth-$q$ expression is given. Generalizing the latter result %
we show that \BP{} is FPT with respect to the \emph{cliquewidth-$q$ 
vertex deletion number}: the number of 
vertices that have to be deleted in order to obtain a graph of constant 
cliquewidth~$q$.\footnote{To be precise, we need the vertex deletion set 
to be given to obtain an FPT algorithm for this parameter.} To the best of our knowledge this parameter has not been 
considered in the past. The cliquewidth-$q$ deletion number is a generalization 
of several
well-studied graph parameters like vertex cover number~($q=1$)~\cite{ChenKX10}, cluster
vertex deletion number and cograph vertex deletion number ($q=2$)~\cite{CO00},
feedback vertex set number~($q=3$)~\cite{KLL02}, and treewidth-$t$ vertex deletion 
number~($q=2^{t + 1} +1$)~\cite{CO00,FLMS12}. 

\subsection{The \VBP{} problem}

The next problem we consider is the \VBP{} problem, for which vertices instead 
of edges need to be removed in order to bisect the graph. More formally, 
let~$G$ be a graph and $S \subseteq V(G)$ be a subset of the vertices of $G$. 
We call~$S$ an \emph{\cvsep{A}{B}} for~$G$ if there are vertex sets~$A, B 
\subseteq V(G)$ such that~$\{S, A, B\}$ forms a partition of~$V(G)$, and there 
are no edges between~$A$ and~$B$ in~$G$. Moreover, we call~$S$ \emph{balanced} 
if~$||A| - |B|| \leq 1$. The problem then is the following.

\decprob{\VBP{}}{A graph~$G$ and a positive integer~$k$.}{Does~$G$ contain a
balanced separator of size at most~$k$?}

We show that this 
problem is more general than \BP{} in the sense that any solution to \VBP{} can 
be transformed into a solution to \BP{} having (almost) the same cut size and 
(almost) the same number of cut out connected components, in polynomial time.
In contrast to \BP{} however, we prove that \VBP{} is W[1]-hard with respect to 
the cut size. In fact this still holds true when combining the cut size and the 
number of cut out connected components as a parameter. This means that to obtain 
a fixed-parameter algorithm it is unavoidable to impose some additional 
constraint. 

\looseness=-1 We show that the \VBP{} problem is FPT with respect to the cut size,
if an optimal solution cuts the graph into a given \emph{constant} number of 
connected components.
We chose this condition as a natural candidate: %
First, in practice optimal bisections often cut into very few connected 
components, typically only into two or three~\cite{Arbenz,KarypisK1998, Werneck, 
DFGRW13}. Second, also for random regular 
graphs the sets~$A$ and~$B$ of the optimum bisection are connected with high 
probability~\cite{BuiCLS87}. And third, \VBP{} remains NP-hard even if all 
optimum bisections cut into exactly two connected components (this follows 
easily by combining the NP-hardness proof for \BP{} of \citet{GareyJS76} with 
our techniques from \autoref{sec:redbptovbp}; see also~\cite{FleischerB09} 
for a related problem). 
To achieve our FPT result for \VBP{}, we generalize the \emph{treewidth 
reduction} technique for separation problems that has been recently introduced 
by \citet{MOR13}. By adapting it to the global balancedness constraint of our 
problem, we address an open question by \citet{MOR13} of whether this is 
possible. 

\subsection{The \BB{} problem}

Apart from \BP{} and \VBP{} we also study the \BB{} problem.
This is a natural generalization of the \BP{} 
problem, in which the $n$ vertices of a graph need 
to be partitioned into~$d$ equal-sized parts, for some arbitrary given number 
$d$ (instead of only two). More formally the problem is defined as follows, where the 
\emph{cut size} of a partition is the number of edges incident to vertices of 
different parts. %

\pagebreak[3]
\decprob{\BB{}}{A graph $G$ and two positive integers $k$ and $d$.}{Is there a 
partition of the vertices of $G$ into $d$ sets of size at most $\lceil 
n/d\rceil$ each and with cut size at most $k$?}

Note that there is no explicit lower bound on the part sizes. This means that 
they can technically speaking be unbalanced. However the definition above is the 
most commonly used one in the literature. Also, for numerous applications such 
as parallel computing~\cite{ArbenzLMMS07} or VLSI circuit 
design~\cite{BhattL84}, only an upper bound is needed.

Our algorithms for the special case \BP{} can easily be extended to 
algorithms for \BB{}, as we will describe in \autoref{appendix:W-hard}. However the algorithms have 
additional running time factors in the order of~$n^{O(d)}$. We observe that \BB{} is W[1]-hard for the number~$d$ of cut out parts even on forests of 
maximum degree two and hence it is unlikely that running time factors of~$n^{f(d)}$ can be avoided. Furthermore, we show that the problem remains 
W[1]-hard on more general graphs for the larger number~$c$ of cut out connected components. 

Regarding structural graph parameters, many of the known hardness results for \BB{} already rule out FPT algorithms for parameters such as treewidth or 
cluster vertex deletion number (\BB{} is NP-hard for trees~\cite{AFFoschini12} 
and graphs formed by a disjoint union of cliques~\cite{RackeA06}). On the 
positive side we can show that \BB{} is FPT with respect to the vertex cover 
number $\tau$. Recently \citet{GanianO13} developed an algorithmic framework 
with which they were able to show that \BB{} is FPT with respect to the combined 
parameters $\tau$ and $d$. Hence we improve on this result by removing the 
dependence on $d$.

\subsection{Organization of the article} 
We begin with presenting our results for \VBP{} in \autoref{sec:constantcomp}. 
These include the hardness of \VBP{}, the FPT algorithm in case the number of 
cut out components is constant, and the reduction from \BP{} showing that \VBP{} 
is more general. \autoref{sec:incomp} contains incompressibility results for 
\BP{}. In \autoref{dcc} we give our FPT algorithm for the cliquewidth-$q$ 
deletion number. Hardness results for \BB{} are given in 
\autoref{appendix:W-hard}, and \autoref{appendix:BB-FPT} contains the FPT 
algorithm w.r.t. the vertex cover number. 

\section{\VBP{} and the Cut Size Parameter}\label{sec:constantcomp}
In this section we show how to compute optimal bisections that cut into some 
constant number of connected components in FPT-time with respect to the 
cut size. As mentioned in the introduction, for
\VBP{} one searches for a small set of vertices in order to bisect a given 
graph. We note below that \VBP{} generalizes \BP{} and hence, there is also a 
corresponding algorithm for \BP{}.

The outline of this section is as follows. First, we note that \VBP{} is W[1]-hard with respect to~$k$ and the number~$c$ of cut out components (\autoref{sec:vbp-hard}). Hence, an additional constraint like~$c$ being constant is unavoidable to get 
an FPT-algorithm. We then proceed to show that \VBP{} indeed generalizes \BP{} (\autoref{sec:redbptovbp}). The FPT algorithm with respect to~$k$ and constant number of cut out components for \VBP{} is given in \autoref{sec:algvbp}.

\subsection{Hardness of \VBP{}}\label{sec:vbp-hard}
\newcommand{\MCC}{\textsc{Mul\-ti\-col\-ored Clique}}

In this section we prove that \VBP{} is W[1]-hard with respect to the combination of the desired separator size and the number of cut out components. It follows, that it is in particular W[1]-hard for the 
parameter separator size.%

\begin{theorem}\label{vbswhard}
\VBP{} is W[1]-hard with respect to the combined parameter~$(k,c)$, where 
$k$~is the desired separator size and $c$~is the maximum number of components after 
removing any set of at most $k$~vertices from the 
input graph.%
\end{theorem}
We reduce from the W[1]-hard \textsc{Clique} problem~\cite{DF13}. The 
reduction is an adaption of the one \citet{Mar06} used to show 
W[1]-hardness for the \ClV{} problem. The construction we use is as follows.%

\begin{construction}\label{conswhard}
Suppose we want to construct an instance~$(G', k)$ of \VBP{} from an 
instance~$(G, k)$ of \textsc{Clique}. Without loss of generality, assume that~$k$ is even. The graph~$G'$ is obtained by first 
copying~$G$ and then subdividing every edge, meaning to replace each~$\{u, w\} 
\in E(G)$ by a new \emph{edge vertex}~$v_{u, w}$ and the edges~$\{u, v_{u, w}\}, 
\{v_{u, w}, w\}$. Next, we make~$V(G)$ into a clique in~$G'$ and we 
furthermore add to~$G'$ a disjoint clique~$D$ with~$n + m - k - 2\binom{k}{2}$ 
vertices, where~$n = |V(G)|$ and~$m = |E(G)|$.
Note that the overall number of vertices in~$G'$ is~$2n + 2m - k - 2\binom{k}{2}$.
\end{construction}
Let us prove that \autoref{conswhard} is a parameterized reduction.
\begin{proof}[Proof of \autoref{vbswhard}]
Assume that~$G$ contains a clique~$C$ of size~$k$. We claim that~$C$ also 
induces a balanced separator in~$G'$. Let~$V^E_C$ be the set of edge vertices 
in~$G'$ corresponding to the edges of~$G[C]$. Consider the sets~$A := V(G') 
\setminus (C \cup V^E_C \cup D)$ and~$B := D \cup V^E_C$. Clearly, $A, B, C$ 
form a partition of~$V(G')$, and since~$C$ is exactly the neighborhood 
of~$V^E_C$ in~$G'$, there are no edges between~$A$ and~$B$. Moreover,
  \begin{align*}
    |A| &= \left(2n + 2m - k - 2\binom{k}{2}\right) - k - \binom{k}{2} - \left(n + m - k - 2 \binom{k}{2}\right)\\
    &= n + m - k - \binom{k}{2} = n + m - k - 2 \binom{k}{2} + \binom{k}{2} = |B|\text{.}
  \end{align*}
Thus, indeed~$C$ is a balanced separator for~$G'$.

For the reverse direction first consider any vertex set~$S \subseteq V(G')$ 
with~$|S| \leq k$. The graph~$G'$ consists of two cliques that, without loss of 
generality, contain more than~$k$ vertices, in addition to degree-two vertices 
attached to the clique on~$V(G)$. Furthermore, any pair of vertices 
in~$V(G)$ has at most one common degree-two neighbor. Hence, the 
number of connected components of~$G' - S$ is at most~$\binom{|S|}{2} + 2$, 
which means that~$c \leq \binom{k}{2} + 2$.
  
Now assume additionally that~$S$ is a balanced \cvsep{A}{B} for~$G'$ and, 
without loss of generality, assume that~$(D \setminus S) \subseteq A$. We may 
furthermore assume that~$S \cap D = \emptyset$. Otherwise we may 
successively replace all vertices in~$S \cap D$ with arbitrary vertices 
from~$(V(G')\setminus D) \cap A$, which is always non-empty since~
\begin{align*}
  |D \setminus S| &\leq |D| = n + m - k - 2\binom{k}{2} \\
  & < \left(2n + 2m - k - 2\binom{k}{2} - k\right)/2 \leq (|V(G')| - |S|)/2
  \text{.}
\end{align*}
Note that, without loss of generality, we may assume further that~$|S|$ is even. Otherwise, $|S| < k$ because $k$~is even by assumption and we may simply add an arbitrary vertex from~$A$ or~$B$ to~$S$. Hence, since~$|V(G')|$ is also even and~$||A| - |B|| \leq 1$, we have~$|A| = |B|$. Thus, in addition to the vertices in~$D$ the set~$A$ needs to get at least
\begin{multline*}
  |V(G') \setminus S|/2 - |D|\\
\geq \left(2n + 2m - 2k - 2 \binom{k}{2}\right)/2 - \left(n + m - k - 2 
\binom{k}{2}\right) = \binom{k}{2}
\end{multline*}
more vertices. Hence, $S$ needs to separate $\binom{k}{2}$~edge vertices from~$V(G)$. This can only be achieved if $|S| \geq k$, and hence $S$~induces a clique of 
size~$k$ in~$G$.
\end{proof}

\subsection{Reducing \BP{} to \VBP{}}\label{sec:redbptovbp}

Before turning to our algorithm for \VBP{} we show that it indeed
transfers also to \BP{}. That is, \VBP{} is more general than \BP{} in
our setting. We say that~$S$ is a \emph{$c$-component separator}
for~$G$ if there are exactly $c$~connected components in~$G - S$. A
\emph{$c$-component bisection} is defined analogously.

\begin{theorem}\label{prop1}
  There is a polynomial-time many-one reduction from \BP{} to \VBP{}
  such that the desired separator size is one larger than the desired cut
  size. Furthermore, each $c$-component bisection for the \BP{}
  instance yields a $(c + 2)$-component balanced separator for the \VBP{} instance and vice versa. 
\end{theorem}
\noindent The basic idea to prove \autoref{prop1} is to subdivide each edge and replace each vertex by a large clique in a given instance of \BP{}. As deleting edge vertices corresponding to cut edges in a bisection then yields imbalanced parts, we use a gadget consisting of two very large cliques connected by a path to rebalance the parts.
\begin{construction}\label{cons:edgetovertex}
  Let~$(G = (V, E), k)$ be an instance of \BP{} and without loss of generality,
  assume that~$k \leq m := |E|$ and denote $n := |V|$. Construct a graph~$G' = (V', 
E')$ as
follows. For each vertex~$v \in V$, introduce a clique~$C_v$ with~$3m + 2$
vertices. Next, for 
every edge~$e = \{u, v\} \in E$ introduce a vertex~$v_e$ and
make~$v_e$ adjacent to all vertices in both~$C_v$ and~$C_u$. Call the set of 
such ``edge vertices''~$V'_E$. Finally, add two cliques~$D_1$ and~$D_2$ with 
$5nm$~vertices each and connect two arbitrary vertices of~$D_1$ and~$D_2$ by a 
path~$P$ with~$m - 1$ inner vertices. Set the desired separator size to~$k + 1$. 
This finishes the construction of a \VBP{} instance~$(G', k + 1)$.
\end{construction}
\noindent We claim that \autoref{cons:edgetovertex} yields a proof for \autoref{prop1}.
\begin{proof}[Proof of \autoref{prop1}]
First, \autoref{cons:edgetovertex} can easily be seen to be doable in 
polynomial time. Let us prove that it is a many-one reduction.

Let~$\{A, B\}$ be a bisection of~$G$, that is $A, B \subseteq V$ such that~$|A| = |B|$
and there are at most $k$~edges between~$A$ and~$B$ in~$G$; call the set of these edges~$S$. Let~$A' = \bigcup_{v \in A}C_v\cup\{v_e \in V_E \mid e \subseteq A\}$, $B' = \bigcup_{v \in B} C_v \cup \{v_e \in V_E \mid e \subseteq B\}$, and $S' := \{v_e \mid e \in S\}$. Note that~$S'$ is an \cvsep{A'}{B'} of $G'[V'_E \cup \bigcup_{v\in V}C_v]$. Furthermore $-m \leq |A'| - |B'| \leq m$. Observe also that for any~$-m \leq i \leq m$ there is a vertex~$v$ on~$P$ such that~$\{v\}$ is a \cvsep{A''}{B''} of~$G'[D_1 \cup D_2 \cup P]$ and~$|A''| - |B''| \in \{i, i + 1\}$. Hence, choosing~$v$ on~$P$ appropriately yields a balanced separator~$S \cup \{v\}$ for~$G'$ that has 
size~$k + 1$. Further, if~$\{A, B\}$ is a $c$-component bisection, then~$S \cup \{v\}$ is a $(c + 2)$-component balanced separator.

 For the converse direction, let~$S$ be a balanced \cvsep{A}{B} 
for~$G'$. Observe that, for the cliques~$D_1, D_2$ and for each of the cliques~$C_v$, $v \in V$, removing~$S$ from their vertex set yields a set which is completely 
contained either in~$A$ or in~$B$. Furthermore, not both~$D_1 \setminus S$ 
and~$D_2 \setminus S$ are contained in~$A$ or in~$B$ since, otherwise, this 
would contradict the balancedness of~$S$. Let us assume without loss of 
generality, that~$D_1 \setminus S \subseteq A$ and~$D_2 \setminus S \subseteq 
B$.
We claim that the number~$a$ of cliques~$C_v$, $v \in V$, that intersect~$A$ is
the same as the number~$b$ of cliques~$C_v$, $v \in V$, that intersect~$B$. 
By the above observation, all vertices of a clique~$C_v$ not contained in~$S$ 
are either completely contained in~$A$ or~$B$. 
Hence, as~$A$ consists of~$D_1 \setminus S$, its intersection with the 
cliques~$C_v$, at most~$m$ edge-vertices~$v_e \in V'_E$, and at most~$m$ vertices 
from~$P$, we have~$$5nm + (3m + 2)
\cdot a + 2m \geq |A| \geq 5nm + (3m + 2) \cdot a - k \geq 5nm + (3m +
2) \cdot a - m$$ and analogously for~$|B|$. Without loss of generality, we may 
assume~$|A| \geq |B|$. Using the above size bounds for~$|A|$ and~$|B|$, we thus 
obtain that~$|A| - |B|$ is at least~$$5nm + (3m + 2) \cdot a - m - (5nm + (3m + 
2) \cdot b + 2m) = (3m + 2) \cdot (a - b) - 3m\text{,}$$ and, thus,~$a = b$ 
(recall that~$S$ is a balanced separator, and hence $||A| - |B|| \leq 1$). 
Since~$D_1 \setminus S \subseteq A$ and~$D_2 \setminus S \subseteq B$ we have at 
least one vertex in~$S \cap P$, and thus the number of edge vertices~$v_e$ in
the balanced separator~$S$ is at most~$k$. We conclude that cutting the edges
according to the edge vertices in~$S$ yields a bisection for~$G$ of cut size at 
most~$k$. Here, too, the bound on the number of connected components is easy to see.
\end{proof}

\subsection{An FPT Algorithm for Cut Size and Constant Number of Cut Out
Components}\label{sec:algvbp}

We now outline an FPT algorithm for \VBP{}. We say that~$S$ is an \cvsep{s}{t} 
for vertices~$s, t$ if there are vertex sets~$A, B \subseteq V(G)$ such that~$S$ 
is an \cvsep{A}{B} and~$s \in A$ and~$t \in B$. We say that an \cvsep{s}{t}~$S$ 
is \emph{inclusion-wise minimal}, or just \emph{minimal}, if there is no 
\cvsep{s}{t}~$S' \subsetneq S$.
We first observe that a balanced separator consists of inclusion-wise minimal 
\cvsep{s}{t}s between a collection of ``terminal'' vertices~$s, t$. The terminal 
vertices are chosen one from each of the connected components of the graph 
without the separator. Guessing the terminals, we can reduce \VBP{} to finding 
an ``almost balanced'' separator consisting of vertices contained in 
inclusion-wise minimal separators of pairs of terminals. To find such an almost 
balanced separator, we generalize the ``treewidth reduction'' technique 
introduced by~\citet{MOR13}. We obtain an algorithm that constructs 
a graph~$G'$ that preserves all inclusion-wise minimal separators of size at 
most~$k$ between some given terminals and has treewidth bounded by some 
function~$g(k, c)$, where~$c$ is the number of terminals. Moreover, the algorithm 
runs in time~$f(k, c) \cdot (n + m)$ and also derives a mapping of the vertices 
between the input graph~$G$ and the constructed graph~$G'$ which allows to 
transfer balanced separators of~$G'$ to balanced separators of~$G$. Using this algorithm it then only remains to show that weighted \VBP{} is 
fixed-parameter tractable with respect to the treewidth. Overall, the algorithm 
solving \VBP{} guesses the terminals, reduces the treewidth and then solves the 
bounded-treewidth problem.

\def\wtorso{\ensuremath{\operatorname{atorso}}}
\def\torso{\ensuremath{\operatorname{torso}}}
The main ingredient in our FPT algorithm for \VBP{} is a generalization of the 
treewidth reduction technique of \citet{MOR13} to graphs with vertex 
weights. We aim to construct a graph of bounded treewidth that preserves all 
inclusion-wise minimal \cvsep{s}{t}s of a given size. To this end, we define 
trimmers.
\begin{definition}
  Let~$G = (V, E)$ be a graph, $k$ an integer and~$T \subseteq V$. A~tuple~$(G^*, \phi)$ of a graph~$G^* = (V^*, E^*)$ and a total, surjective, but not necessarily injective mapping~$\phi \colon V \to V^*$ is called a \emph{$(k, T)$\nobreakdash-trimmer} of~$G$ if the following holds. (Here, we let~$\phi^{-1}(v) := \{v' \mid \phi(v') = v\}$, $\phi(V') := \bigcup_{v \in V'} \phi(v)$ for~$V' \subseteq V(G)$, and define~$\phi^{-1}(V')$ analogously.)
  \begin{compactenum}[(i)]
  \item %
    For any~$S \subseteq V^*$, the mapping~$\phi$ is a one-to-one mapping between the connected components of~$G - \phi^{-1}(S)$ and~$G^* - S$.\label{enu:trp1}
  \item If~$S$ is an inclusion-wise minimal \cvsep{s}{t} for~$G$ with~$|S| \leq k$ and~$s, t \in T$, then~$\phi(S) = S$ and~$S$ is an inclusion-wise minimal \cvsep{\phi(s)}{\phi(t)} for~$G^*$.\label{enu:trp2}
  \end{compactenum}
  We refer to the above as \emph{trimmer properties}~\enuref{enu:trp1} and~\enuref{enu:trp2}.
\end{definition}

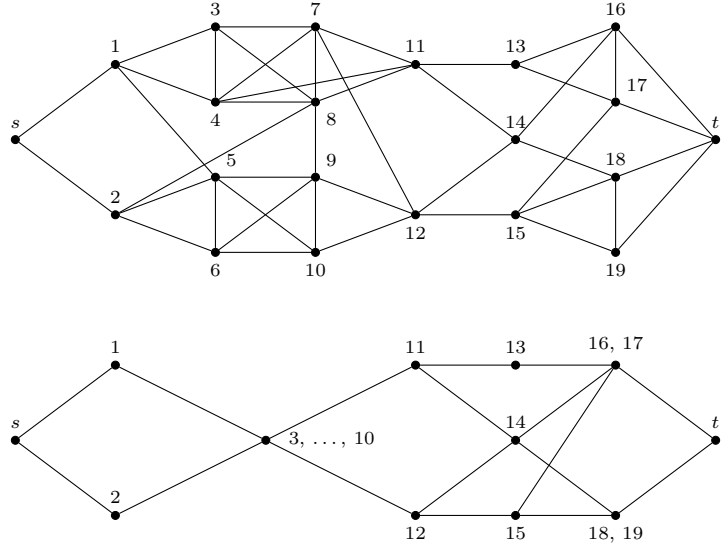
\begin{figure}[t]
  \centering\scriptsize
  \begin{tikzpicture}[node distance=10mm and 13.3mm]
    \tikzstyle{ver}=[circle,draw,fill=black,minimum size=3pt,inner
    sep=0pt, on grid]
    
    \node[ver,label=above:$s$] (s) {};
    \node[ver,above right= of s, label=above:1] (v1) {}
    edge (s);
    \node[ver,below right= of s, label=above:2] (v2) {}
    edge (s);
    \node[ver,above right= of v1, yshift=-5mm, label=above:3] (v3) {}
    edge (v1);
    \node[ver,below= of v3, label=below:4] (v4) {}
    edge (v1)
    edge (v3);
    \node[ver,below= of v4, label=above right:5] (v5) {}
    edge (v1)
    edge (v2);
    \node[ver,below= of v5, label=below:6] (v6) {}
    edge (v2)
    edge (v5);
    \node[ver,right= of v3, label=above:7] (v7) {}
    edge (v3)
    edge (v4);
    \node[ver,right= of v4, label=below right:8] (v8) {}
    edge (v2)
    edge (v3)
    edge (v4)
    edge (v7);
    \node[ver,right= of v5, label=above right:9] (v9) {}
    edge (v5)
    edge (v6)
    edge (v8);
    \node[ver,right= of v6, label=below:10] (v10) {}
    edge (v5)
    edge (v6)
    edge (v9);
    \node[ver,right= of v7, yshift=-5mm, label=above:11] (v11) {}
    edge (v4)
    edge (v7)
    edge (v8);
    \node[ver,right= of v9, yshift=-5mm, label=below:12] (v12) {}
    edge (v7)
    edge (v9)
    edge (v10);
    \node[ver,right= of v11, label=above:13] (v13) {}
    edge (v11);
    \node[ver,below= of v13, label=above:14] (v14) {}
    edge (v11)
    edge (v12);
    \node[ver,below= of v14, label=below:15] (v15) {}
    edge (v12);
    \node[ver,above right= of v13, yshift=-5mm, label=above:16] (v16) {}
    edge (v13)
    edge (v14);
    \node[ver,below= of v16, label=above right:17] (v17) {}
    edge (v13)
    edge (v15)
    edge (v16);
    \node[ver,below= of v17, label=above:18] (v18) {}
    edge (v14)
    edge (v15);
    \node[ver,below= of v18, label=below:19] (v19) {}
    edge (v15)
    edge (v18);
    \node[ver,right= of v17, yshift=-5mm, label=above:$t$] (t) {}
    edge (v16)
    edge (v17)
    edge (v18)
    edge (v19);

    \node[ver,below = 4cm of s, label=above:$s$] (s) {};
    \node[ver,above right= of s, label=above:1] (v1) {}
    edge (s);
    \node[ver,below right= of s, label=above:2] (v2) {}
    edge (s);
    \node[above right= of v1, yshift=-5mm] (v3) {};
    \node[ver,right= 20mm of v2, yshift=10mm, label=right:{\hspace{.5mm} 3, \ldots, 10}] (c1) {}
    edge (v1)
    edge (v2);    
    \node[ver,below= 4cm of v11, label=above:11] (v11) {}
    edge (c1);
    \node[ver,below= 4cm of v12, label=below:12] (v12) {}
    edge (c1);
    \node[ver,right= of v11, label=above:13] (v13) {}
    edge (v11);
    \node[ver,below= of v13, label=above:14] (v14) {}
    edge (v11)
    edge (v12);
    \node[ver,below= of v14, label=below:15] (v15) {}
    edge (v12);
    \node[ver,right= of v13, label=above:{16, 17}] (c2) {}
    edge (v13)
    edge (v14)
    edge (v15);
    \node[ver,right= of v15, label=below:{18, 19}] (c3) {}
    edge (v14)
    edge (v15);
    \node[ver,below right= of c2, label=above:$t$] (t) {}
    edge (c2)
    edge (c3);
    
  \end{tikzpicture}
  
  \caption{A graph~$G$ (top) and the graph in a $(3, \{s, t\})$-trimmer of~$G$ (bottom), computed as we show in \autoref{sec:twred}. The associated mapping~$\phi$ is indicated by the vertex labels.}
\label{fig:ex-trimmer}
\end{figure}
An example for a trimmer is given in \autoref{fig:ex-trimmer}: one can verify by inspecting the graph~$G$ depicted there, that only vertices 1, 2, 11, 12, 13, 14, and~15 are contained in any inclusion-wise minimal \cvsep{s}{t} of size at most~3. For example, vertex~3 is not contained in any such separator because removing it from the graph leaves two vertex-disjoint \pth{s}{t}s and all pairs of such paths can only be destroyed by removing two further vertices if we delete either vertex~1 and~2, or vertex~11 and~12. However, both~$\{1, 2\}$ and~$\{11, 12\}$ are themselves \cvsep{s}{t}s and hence adding~3 does not yield an \emph{inclusion-wise minimal} separator. Thus contracting every edge in~$G[3, \ldots, 10, 16, \ldots, 19]$ we can derive a mapping~$\phi$ that fulfills trimmer property~\enuref{enu:trp2}. Basically, trimmer property~\enuref{enu:trp1} is obtained by observing that contracting edges keeps intact all important paths. 

A more precise description of computing a trimmer, a formal proof of the properties, and an upper bound on the treewidth of the trimmer is given in \autoref{sec:twred}: as we show in \autoref{thm:treewidth-reduction-gen}, if~$k$ and~$|T|$ are small, then there are trimmers of small treewidth that can be computed efficiently.
\newcommand{\twredthm}[1]{
  Let~$G$ be a graph. For every constant $k\in\mathbb N$ and constant-size $T \subseteq V$, we can compute a $(k, T)$\nobreakdash-trimmer~$(G^*, \phi)$ for~$G$ in $O(n + m)$~time such that the treewidth of~$G^*$ is at most~$g(k, |T|)$ for some function~$g$ depending only on~$k$ and~$|T|$.%
}

\begin{theorem}\label{thm:treewidth-reduction-gen}
  \twredthm{}
\end{theorem}%
\noindent The final ingredient for our FPT algorithm for \VBP{} is an efficient algorithm for small treewidth and vertex weights.

\begin{theorem}\label{lem:vbp-weights-tw}
 Let~$G$ be a graph with treewidth $\omega$ and integer vertex-weights~$\wf$. Let~$\WF$ be
the sum of all vertex weights and let~$c \geq 2$ be an integer. We
can find in $\tw^{O(\tw)} \cdot c^2 \cdot
\WF^2\cdot n$~time, for all integers~$1 \leq s \leq \WF$, a partition~$\{A, B, S\}$ of~$V(G)$ such that $\wf(A) = s$ and $S$~is a minimum-weight $c$-component \cvsep{A}{B}, or reveal that no such partition exists.
\end{theorem}
\noindent We defer also this proof until later in \autoref{sec:twalg}. If we suppose for the moment that the above \autoref{thm:treewidth-reduction-gen} and \autoref{lem:vbp-weights-tw} hold, then we arrive at the main theorem of this section.
\begin{theorem}\label{thm2}
  Let~$G$ be a graph. Given non-negative integers~$c$ and~$k$, in~$h(c, k) \cdot n^{c + 3}$~time 
  we can find a $c$-component balanced separator for~$G$ of size at most~$k$ if it exists. Here,~$h(c, k)$ is a function depending only on~$c$ and~$k$.
\end{theorem}
\begin{proof}
 \def\cutop{\ensuremath{\operatorname{cut}}}
The algorithm proceeds as follows.
For each $T \subseteq V(G)$ of size exactly~$c$ we compute a~$(k, T)$-trimmer~$(G^*, 
\phi)$ using \autoref{thm:treewidth-reduction-gen}. We create a vertex weight 
function~$\lambda$ for~$G^*$ by letting~$\lambda(v) = |\phi^{-1}(v)|$. Then, for 
each~$s$,~$|V(G)|/2 - 1 - k \leq s \leq |V(G)|/2 + k$, we compute a 
minimum-weight $c$-component \cvsep{A'}{B'} for~$G^*$ and the corresponding sets~$A'$,~$B'$ with~$\lambda(A') = s$ 
using \autoref{lem:vbp-weights-tw}. If among the separators there is an
\cvsep{A'}{B'}~$S'$ with $|\wf(A') - \wf(B')| \leq k - \wf(S') + 1$, then we 
compute~$S := \phi^{-1}(S'), A := \phi^{-1}(A')$, and~$B := \phi^{-1}(B')$.
Note that, by trimmer property~\enuref{enu:trp1}, $S$~is a $c$-component 
\cvsep{A}{B} for~$G$. Moreover, since~$\phi$ is a total mapping, $||A| - |B|| 
\leq k - |S| + 1$. We move~$k - |S|$ vertices from~$A$ or~$B$ to~$S$ in such a way 
that~$S$ is a $c$-component balanced separator for~$G$ and we output~$S$. If no suitable separator is found, we output that there is no 
$c$-component balanced separator of size at most~$k$ for~$G$. Note 
that, unless~$|V(G)|$ is bounded by a function of~$k$ (i.e. the problem is 
trivially FPT), moving the vertices from~$A$ or~$B$ to~$S$ without changing the number of components 
of~$G - S$ is always possible. This is because not every vertex of a connected 
component can separate it into multiple ones and there is always a component of 
size at least two.

Let~$S$ be a $c$-component balanced separator of size at most~$k$ for~$G$ and pick vertices~$v_1, \ldots, v_c$, one from each connected component of~$G - S$. Let us observe that the
above algorithm finds a $c$-component balanced separator of size at most~$k$. Note that~$S$ is a
\cvsep{v_i}{v_j} for each~$1 \leq i < j \leq c$. Hence, $S$ contains
inclusion-wise minimal \cvsep{v_i}{v_j}s~$S_{i, j}$ of size at most~$k$. Let~$\hat{S}
= \bigcup_{1 \leq i < j \leq c} S_{i, j}$, call a connected component in~$G - \hat{S}$ \emph{odd} if it does not contain any~$v_i$, and let $\tilde{S}$ be the union of $\hat{S}$ and all odd components. Note that odd components are contained in~$S$. Hence, $\tilde{S}$ is a $c$-component \cvsep{\tilde{A}}{\tilde{B}} for~$G$ with~$||\tilde{A}| - |\tilde{B}|| \leq k - |\tilde{S}| + 1$ and $|V(G)|/2 - 1 - k \leq |\tilde{A}| \leq |V(G)|/2 + k$. Clearly, at some point in the algorithm~$T = \{v_1, \ldots, v_c\}$ and the $(k, |T|)$\nobreakdash-trimmer~$(G^*, \phi)$ of~$G$ is computed. It remains to show that separator~$\tilde{S}$ induces a separator in~$G^*$ that is found by the algorithm. By trimmer property~\enuref{enu:trp2} we have that~$\phi(\hat{S}) = \hat{S}$ is contained in~$G^*$. Trimmer property~\enuref{enu:trp1} %
gives that $\phi$ is a one-to-one mapping of connected components~$C$ in~$G - \hat{S}$ and their counterparts~$\phi(C)$ in~$G^* - \hat{S}$. In particular, there is a corresponding mapping for all odd connected components. Thus, $\phi(\tilde{S})$ is a $c$-component \cvsep{\phi(\tilde{A})}{\phi(\tilde{B})} for~$G^*$ and we have $\lambda(\phi(\tilde{S}))=|\tilde{S}|, \lambda(\phi(\tilde{A}))=|\tilde{A}|$, and $\lambda(\phi(\tilde{B}))=|\tilde{B}|$. 
Hence, an \cvsep{A'}{B'}~$S'$ for~$G^*$ with~$\wf(S') \leq \wf(\phi(\tilde{S}))$ and $\wf(A') = \wf(\phi(\tilde{A}))$ is enumerated by the algorithm of 
\autoref{lem:vbp-weights-tw}. Applying the size bounds of~$\tilde{A}, \tilde{B}, 
\tilde{S}$ we have $|\wf(A') - \wf(B')| \leq k - \wf(S') + 1$ and 
$|V(G)|/2 - 1 - k \leq \wf(A') \leq |V(G)|/2 + k$. Thus, the algorithm 
described above finds a $c$-component balanced separator of size at most~$k$ 
for~$G$. 

Concerning the running time, there are at most~$n^c$ computations of the 
trimmer, each of which can be done in~$f(k, c) \cdot (n + 
m)$~time~(\autoref{thm:treewidth-reduction-gen}), where $m = |E(G)|$. Then we compute the separators 
for~$G^*$ and since the treewidth of~$G^*$ is bounded by some function~$g(k, 
c)$ (\autoref{thm:treewidth-reduction-gen}), this can be done in 
time~$g(k,c)^{O(g(k, c))} \cdot n^3$ using \autoref{lem:vbp-weights-tw}. Next, the parts~$A', B'$ can be computed in linear time from the parts~$A, B$ and 
we can modify the algorithm for \autoref{lem:vbp-weights-tw} to also output~$A, 
B$ without increasing the running time bound (see \autoref{sec:twalg}). Finally, moving the vertices from the parts~$A$ or~$B$ to~$S$ can be done in~$O(k(n + m))$~time because we move at most~$k$ vertices and a vertex that does not change the number of components can be found in~$O(n + m)$~time by taking a leaf of a BFS tree of a component that contains at least two vertices. Hence, the overall running time is bounded by~$n^c \cdot (f(k, c) \cdot (n + 
m) + g(k, c)^{O(g(k, c))} \cdot n^3 + O(n + m) + O(k(n + m)))$ which in turn is bounded by~$h(c, k) \cdot 
n^{c + 3}$ for a suitable function~$h$.
\end{proof}

We remark that an upper bound on the treewidth of the trimmer is~$2^{O(k^2)}$ 
(see \autoref{rem:upper-lower-bound-treewidth} below), yielding 
a~$2^{2^{O(k^2)}\cdot c^5} \cdot n^{c + 3}$-time algorithm. This bound can 
most certainly be improved, and it would be interesting to know what kinds of 
lower bounds exist.
Applying \autoref{prop1} we can derive the following.
\begin{corollary}
  Let~$G$ be a graph. Given non-negative integers~$c$ and~$k$, in~$h(c, k) \cdot n^{c + 9}$~time 
  we can find a $c$-component bisection for~$G$ of size at most~$k$ if it exists. Here,~$h(c, k)$ is a function depending only on~$c$ and~$k$.
\end{corollary}

Note that a direct application of \autoref{prop1} yields only a factor 
of~$n^{3(c + 5)}$ in the running time. The trick is, however, that two of the 
terminals in instances created by the corresponding \autoref{cons:edgetovertex} 
do not need to be guessed, while for the others we only have to guess in which 
vertex clique they appear. They can be assumed to be arbitrary vertices in the 
appropriate cliques. We omit the straightforward details. We remark that the 
FPT algorithm given by \citet{CLPPS14} has a much better singly exponential 
runtime, even if $c$ is constant. However, as said above, we believe that the 
runtime of our algorithm can be significantly improved. Hence our observations 
might still yield faster FPT algorithms for constant~$c$ in instances from practice. As described in the 
introduction, this would have implications for applications, as there 
$c$ often is a small constant.

Let us now deliver the proofs of \autoref{thm:treewidth-reduction-gen} and \autoref{lem:vbp-weights-tw}.

\subsubsection{Treewidth Reduction}\label{sec:twred}
In order to prove \autoref{thm:treewidth-reduction-gen}, we need to generalize the results of \citet{MOR13}, which we
do using the following definition. 

\begin{definition}
Let $G = (V, E)$ be a graph. %
The \emph{annotated torso~$\wtorso(G, W)$}, for a set~$W \subseteq V$, is a tuple~$(G', \phi)$ of a graph~$G' = (V', E)$ and a total, surjective, but not necessarily injective mapping~$\phi \colon V \to V'$ defined as follows. The graph~$G'$ is obtained from~$G$ by contracting all edges that have empty intersection with~$W$ and removing all loops and parallel edges created in the process. Hence, each connected component of~$G - W$ has a corresponding vertex in~$G'$ which we call \emph{component vertex}. The mapping~$\phi$ is defined as the identity when restricted to~$W$, that is, for all~$v \in W$ we have~$\phi(v) = v$. For all remaining vertices~$v \in V \setminus W$, $\phi$ maps $v$ to the component vertex of the connected component of~$G - W$ which contains~$v$.

The \emph{torso~$\torso(G, W)$} is obtained from~$G'$ in~$(G', \phi) = \wtorso(G, W)$ as follows. First, make the neighborhood of each component vertex into a clique, creating \emph{shortcut edges}. Then remove all component vertices.\footnote{This definition of torso is equivalent to the one used by \citet{MOR13}.}
\end{definition}
Some graphs and their (annotated) torso graphs are depicted in \autoref{fig:ex-torso}.
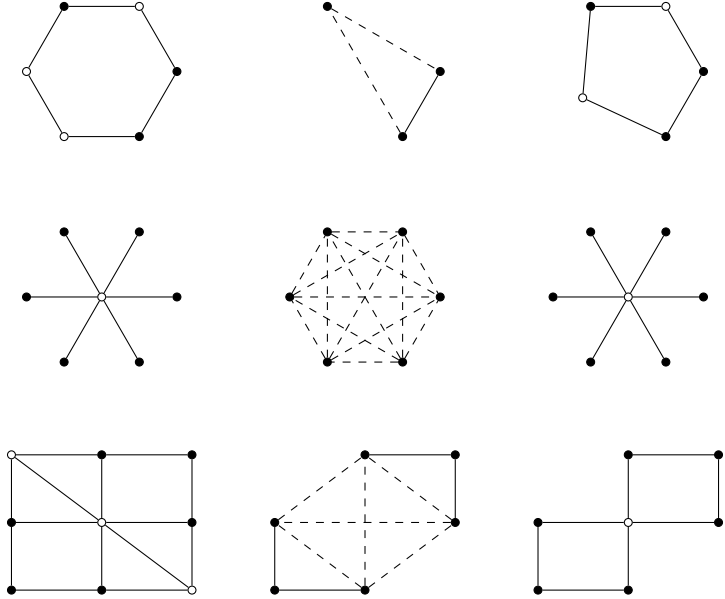
\begin{figure}[t]
  \centering
  \begin{tikzpicture}[node distance=9mm and 12mm]
    \tikzstyle{ver}=[circle,draw,fill=black,minimum size=3pt,inner
    sep=0pt, on grid]
    \tikzstyle{wver}=[circle,draw,fill=white,minimum size=3pt,inner
    sep=0pt, on grid]
    \tikzstyle{shortcut}=[dashed]

    \node[ver] (v2) at (0:10mm)  {};
    \node[wver] (v3) at (60:10mm)  {}
    edge (v2);
    \node[ver] (v4) at (120:10mm)  {}
    edge (v3);
    \node[wver] (v5) at (180:10mm)  {}
    edge (v4);
    \node[wver] (v6) at (240:10mm)  {}
    edge (v5);
    \node[ver] (v7) at (300:10mm)  {}
    edge (v6)
    edge (v2);

    \begin{scope}[shift={(3.5cm, 0)}]
      \node[ver] (v2) at (0:10mm) {}; 
      \node[ver] (v4) at (120:10mm) {} 
      edge[shortcut] (v2);
      \node[ver] (v7) at (300:10mm) {} 
      edge[shortcut] (v4) 
      edge (v2);
    \end{scope}
    
    \begin{scope}[shift={(7cm, 0)}]
      \node[ver] (v2) at (0:10mm)  {};
      \node[wver] (v3) at (60:10mm)  {}
      edge (v2);
      \node[ver] (v4) at (120:10mm)  {}
      edge (v3);
      \node[wver] (v5) at (210:7mm)  {}
      edge (v4);
      \node[ver] (v7) at (300:10mm)  {}
      edge (v5)
      edge (v2);     
    \end{scope}

    \begin{scope}[shift={(0,-3cm)}]
      \node[wver] (v1) at (0,0)  {};
      \node[ver] (v2) at (0:10mm)  {}
      edge (v1);
      \node[ver] (v3) at (60:10mm)  {}
      edge (v1);
      \node[ver] (v4) at (120:10mm)  {}
      edge (v1);
      \node[ver] (v5) at (180:10mm)  {}
      edge (v1);
      \node[ver] (v6) at (240:10mm)  {}
      edge (v1);
      \node[ver] (v7) at (300:10mm)  {}
      edge (v1);      
    \end{scope}

    \begin{scope}[shift={(3.5cm,-3cm)}]
      \node[ver] (v2) at (0:10mm)  {};
      \node[ver] (v3) at (60:10mm)  {}
      edge[shortcut] (v2);
      \node[ver] (v4) at (120:10mm)  {}
      edge[shortcut] (v2)
      edge[shortcut] (v3);
      \node[ver] (v5) at (180:10mm)  {}
      edge[shortcut] (v2)
      edge[shortcut] (v3)
      edge[shortcut] (v4);
      \node[ver] (v6) at (240:10mm)  {}
      edge[shortcut] (v2)
      edge[shortcut] (v3)
      edge[shortcut] (v4)
      edge[shortcut] (v5);
      \node[ver] (v7) at (300:10mm)  {}
      edge[shortcut] (v2)
      edge[shortcut] (v3)
      edge[shortcut] (v4)
      edge[shortcut] (v5)
      edge[shortcut] (v6);
    \end{scope}

    \begin{scope}[shift={(7cm,-3cm)}]
      \node[wver] (v1) at (0,0)  {};
      \node[ver] (v2) at (0:10mm)  {}
      edge (v1);
      \node[ver] (v3) at (60:10mm)  {}
      edge (v1);
      \node[ver] (v4) at (120:10mm)  {}
      edge (v1);
      \node[ver] (v5) at (180:10mm)  {}
      edge (v1);
      \node[ver] (v6) at (240:10mm)  {}
      edge (v1);
      \node[ver] (v7) at (300:10mm)  {}
      edge (v1);      
    \end{scope}

    \begin{scope}[shift={(0,-6cm)}]
      \node[wver] (v1) at (0,0) {};
      \node[wver, above left= of v1] (v2) {}
      edge (v1);
      \node[ver, above= of v1] (v3) {}
      edge (v1)
      edge (v2);
      \node[ver, above right= of v1] (v4) {}
      edge (v3);
      \node[ver, left= of v1] (v5) {}
      edge (v1)
      edge (v2);
      \node[ver, right= of v1] (v6) {}
      edge (v1)
      edge (v4);
      \node[ver, below left= of v1] (v7) {}
      edge (v5);
      \node[ver, below= of v1] (v8) {}
      edge (v1)
      edge (v7);
      \node[wver, below right= of v1] (v9) {}
      edge (v1)
      edge (v6)
      edge (v8);
    \end{scope}

    \begin{scope}[shift={(3.5cm,-6cm)}]
      \node (v1) at (0,0) {};
      \node[ver, above= of v1] (v3) {};
      \node[ver, above right= of v1] (v4) {}
      edge (v3);
      \node[ver, left= of v1] (v5) {}
      edge[shortcut] (v3);
      \node[ver, right= of v1] (v6) {}
      edge[shortcut] (v3)
      edge[shortcut] (v5)
      edge (v4);
      \node[ver, below left= of v1] (v7) {}
      edge (v5);
      \node[ver, below= of v1] (v8) {}
      edge[shortcut] (v3)
      edge[shortcut] (v5)
      edge[shortcut] (v6)
      edge (v7);
    \end{scope}

    \begin{scope}[shift={(7cm,-6cm)}]
      \node[wver] (v1) at (0,0) {};
      \node[ver, above= of v1] (v3) {}
      edge (v1);
      \node[ver, above right= of v1] (v4) {}
      edge (v3);
      \node[ver, left= of v1] (v5) {}
      edge (v1);
      \node[ver, right= of v1] (v6) {}
      edge (v1)
      edge (v4);
      \node[ver, below left= of v1] (v7) {}
      edge (v5);
      \node[ver, below= of v1] (v8) {}
      edge (v1)
      edge (v7);
    \end{scope}
  \end{tikzpicture}
  
  \caption{Graphs~$G$ and vertex subsets~$W$ (black vertices) on the left, $\torso(G, W)$ in the middle (new shortcut edges are dashed), and the graph in $\wtorso(G, W)$ on the right (component vertices are white).}
\label{fig:ex-torso}
\end{figure}
We begin with an observation about annotated torsos and their connected components if some set of vertices is removed.
\begin{lemma}\label{pro:wtorso-sep}
  Let $G$ be a graph, $W \subseteq V(G)$, and $(G', \phi) = \wtorso(G, W)$. For any~$S \subseteq V(G')$, the mapping~$\phi$ is a one-to-one mapping between the connected components of~$G - \phi^{-1}(S)$ and~$G' - S$.

\end{lemma}
\begin{proof}
  We prove that~$\phi$ maps connected components of $G - \phi^{-1}(S)$ to connected components of~$G' - S$ and then prove that it is indeed a one-to-one mapping. Before proving the first part, we note that \pth{s'}{t'}s in~$G' - S$ translate to \pth{s}{t}s in~$G - \phi^{-1}(S)$ for all~$s \in \phi^{-1}(s'), t \in \phi^{-1}(t')$ and vice versa.
  \begin{mathenum}
  \item If~$v_1, \ldots, v_\ell$ is a path~$P$ in~$G - \phi^{-1}(S)$, then~$\phi(v_1), \ldots, \phi(v_\ell)$ is a walk~$P'$ in~$G' - S$ where multiple consecutive occurrences of a vertex in the second sequence are omitted. Hence there is a~$\phi(v_1)$-$\phi(v_\ell)$-path in~$G' - S$. \label{me:path1}
  \suspend{mathenum}
It is enough to show that~$\phi$ maps each pair $x,y$ of adjacent vertices either to the same vertex or to adjacent vertices.
This is clear if both $x,y$ are in $W$. If both $x$ and $y$ are in $V(G)\setminus W$ then they are mapped to the same component vertex of $G'$, as they are in the same connected component of $G - W$. Finally, if $x$ in $W$ and $y \in V(G)\setminus W$ or vice versa, then $\phi(x)$ and $\phi(y)$ are adjacent by the definition of annotated torso. 

   \resume{mathenum}
  \item If~$v_1, \ldots, v_\ell$ is a path~$P'$ in~$G' - S$, then there is an~\pth{s}{t}~$P$ in~$G - \phi^{-1}(S)$ for every~$s \in \phi^{-1}(v_1), t \in \phi^{-1}(v_\ell)$.\label{me:path2}
  \suspend{mathenum}
  Construct~$P'$ as follows. First, consider a component vertex~$v_i \in V(P')$ with~$1 < i < \ell$, if there is any. Note that~$v_{i - 1}$ and~$v_{i + 1}$ are not component vertices, because component vertices are not adjacent with each other.
  We know that~$\phi^{-1}(v_i)$ is a connected component in~$G - W$ and, since~$v_i \notin S$, also~$\phi^{-1}(S) \cap \phi^{-1}(v_i) = \emptyset$. Since $v_i$ is a component vertex adjacent to~$v_{i-1}$ and~$v_{i + 1}$, there are~$v'_{i-1}$ and~$v'_{i + 1}$ in~$\phi^{-1}(v_i)$ adjacent to~$v_{i-1}$ and~$v_{i + 1}$, respectively.  
  Modify~$P'$ by replacing~$v_i$ with a \pth{v_{i - 1}}{v_{i + 1}} in~$G' - \phi^{-1}(S)$ formed by the edges~$v_{i-1}v'_{i-1}$, $v'_{i+1}v_{i + 1}$ and a $v'_{i-1}$-$v'_{i + 1}$-path inside the component~$\phi^{-1}(v_i)$.
  
  Next, if~$v_1$ is a component vertex, choose an arbitrary \pth{s}{v_{2}}~$\hat{P'}$ in~$G' - \phi^{-1}(S)$. Such a path exists by a similar argument as above. Replace~$v_1$ with~$\hat{P'}$ in~$P'$. Proceed analogously if~$v_\ell$ is a component vertex. Since, in this way, we replaced all vertices in~$V(P') \setminus W$ with paths that exist in~$G - \phi^{-1}(S)$, we have obtained an $s$-$t$-walk in~$G - \phi^{-1}(S)$. 
  Hence also \enuref{me:path2} is proved.

  For a graph~$G$ and a set $T \subseteq V(G)$, let us call a set $X\subseteq V(G) \setminus T$ of vertices \emph{$T$-unbroken} if there is a path in~$G-T$ between any two vertices in~$X$. Note that a connected component of~$G-T$ is an inclusion-wise  maximal $T$-unbroken set. Claims \enuref{me:path1} and \enuref{me:path2} show that if~$C$ is $(\phi^{-1}(S))$-unbroken in~$G$, then~$\phi(C)$ is $S$-unbroken in~$G'$ and if~$C'$ is $S$-unbroken in~$G'$, then~$\phi^{-1}(C)$ is $(\phi^{-1}(S))$-unbroken in~$G$. 
  
   Now let us prove that $\phi$ indeed provides a mapping between the connected components. 
  \resume{mathenum}
  \item If a set~$C \subseteq V(G)$ is a connected component in~$G - \phi^{-1}(S)$ then~$\phi(C)$ is a connected component in~$G' - S$.\label{me:mapping}
  \suspend{mathenum}
  If~$C$ is a connected component in~$G - \phi^{-1}(S)$ then~$\phi(C)$ is $S$-unbroken in~$G'$. Now for the sake of contradiction assume that there is a connected component~$K' \supsetneq \phi(C)$ in~$G'-S$. Then~$\phi^{-1}(K')$ is $(\phi^{-1}(S))$-unbroken in~$G$ and~$\phi^{-1}(K') \supsetneq C$ as $\phi$ is surjective. This contradicts~$C$ being a connected component. Hence~$\phi(C)$ is a connected component of~$G' - S$. 
  
   Finally, let us prove that the mapping is one-to-one. 
  \resume{mathenum}
  \item If a set~$C' \subseteq V(G')$ is a connected component in~$G' - S$ then~$\phi^{-1}(C')$ is a connected component in~$G -\phi^{-1}(S)$.\label{me:onetoone}
  \end{mathenum}
  If~$C'$ is a connected component in~$G' - S$ then~$\phi^{-1}(C')$ is $(\phi^{-1}(S))$-unbroken in~$G$. Now for the sake of contradiction assume that there is a connected component~$K \supsetneq \phi^{-1}(C')$ in~$G-\phi^{-1}(S)$. Then~$\phi(K)$ is $S$-unbroken in~$G'$ and~$\phi(K) \supsetneq C'$ by the definition of~$\phi^{-1}$ as~$K \supsetneq \phi^{-1}(C')$. This contradicts~$C'$ being a connected component. Hence~$\phi^{-1}(C')$ is a connected component of~$G-\phi^{-1}(S)$.  
\end{proof}

We now show that the treewidth of an annotated torso is at most one larger than the treewidth of the corresponding torso. For this we need to formally introduce the treewidth
first.
A \emph{tree decomposition} of a graph~$G=(V,E)$ is a pair~$(T, \tau)$, where $T$
is a rooted tree and $\tau$ is a mapping~$V(T) \to 2^V$ such that
\begin{itemize}
\item for every $e \in E(G)$, there is an $x \in V(T)$ with $e
  \subseteq \tau(x)$, and
\item for every $v \in V$ the set $T(v)= \{x \in
  V(T) \mid v \in \tau(x)\}$ induces a non-empty subtree of $T$.
\end{itemize}
The sets~$\tau(x)$ are sometimes called \emph{bags}. The \emph{width} of a tree decomposition~$(T, \tau)$ is
$\max\{|\tau(x)| \mid x \in V(T)\} -1$ and the \emph{treewidth}~$\tw(G)$ of a
graph~$G$ is the minimum width of a tree decomposition for~$G$.

\begin{lemma}\label{pro:wtorso-tw}
 Let~$G = (V, E)$ be a graph, let~$W \subseteq V(G)$ and
let~$\torso(G, W)$ be of treewidth~$\tw$. Then, the graph in~$\wtorso(G, W)$ has treewidth at
most~$\tw + 1$.
\end{lemma}
\begin{proof}
 Let~$(T, \tau)$ be a tree-decomposition for~$\torso(G, W)$. Note that, to
obtain a tree-decomposition for~$\wtorso(G, W)$ we only need to incorporate the
component vertices and their incident edges. For this, successively consider
each component vertex~$v$ and let~$N$ be its neighborhood in~$\wtorso(G, W)$.
Due to the shortcut edges between neighbors of $v$,~$N$ induces a clique
in~$\torso(G, W)$. Hence there is a bag~$\tau(t)$,~$t \in V(T)$, that contains~$N$~\cite[Lemma 2.2.2.]{Klo94}.
Add a new vertex~$t'$ to~$T$ adjacent only to~$t$ and define~$\tau(t') = \tau(t)
\cup \{v\}$. Hence, the bags containing~$v$ or~$N$, respectively, (still) induce
a subtree in~$V(T)$. Furthermore, each edge incident to~$v$ is contained
in~$\tau(t')$. Note that we never need to introduce a copy of a bag that does not occur in the tree-decomposition of~$\torso(G, W)$ because component vertices are not adjacent with each other. Thus, introducing a new bag for all component vertices in the above-described way yields a
tree-decomposition for~$\wtorso(G, W)$ and increases the maximum bag size by at
most one.
\end{proof}
\begin{remark}
  We mention in passing that the notion of annotated torso is more robust with respect to the treewidth than the notion of torso introduced by \citet{MOR13} in the following sense. Adding vertices to the set~$W$ increases the treewidth of both $\torso(G,W)$ and the graph in $\wtorso(G, W)$ by at most one, as \autoref{cor:torso-add-tw} (below) shows. However, removing a vertex from~$W$ may increase the treewidth of $\torso(G,W)$ arbitrarily. This is witnessed by a star where~$W$ contains all vertices; removing the center of the star from~$W$ yields a clique in the torso (see \autoref{fig:ex-torso}). In contrast, the treewidth of annotated torsos cannot increase when removing a vertex~$v$ from~$W$. This is because, either the graph in the annotated torso stays the same, or removing~$v$ is equivalent to contracting an edge between~$v$ and a component vertex. That is, the resulting graph is a minor of the original annotated torso graph.
\end{remark}

As we observe next, annotated torsos can be computed in linear time.

\begin{lemma}\label{prop:compute-torso}
Let~$G = (V, E)$ be a graph and let~$W \subseteq V$.
Then, %
$(G', \phi) = \wtorso(G, W)$ and~$\phi^{-1}$ can be computed in $O(n + m)$~time. %
\end{lemma}
\begin{proof}
To compute the annotated torso, color the connected
components of~$G - W$ with distinct colors~$1, \ldots, c$. %
Next, create a copy~$G'$ of~$G$ and modify~$G'$ by contracting all edges that have empty intersection with~$W$ and removing all self-loops and parallel edges. In order to achieve a linear time bound for this, we proceed as follows (we assume~$G'$ to be represented as an adjacency list structure).

We start with vertices in $W$ having their adjacency lists as in $G$ and introduce new vertices~$v_1, \ldots, v_c$ (the component vertices) with empty lists. For each $w$ in $W$ we process its adjacency list and for each vertex $x$ on the list we distinguish the following cases:
\begin{itemize}
 \item [Case 1:] the vertex~$x$ is in~$W$. We make no changes to the adjacency lists and proceed with the next vertex in the adjacency list.
 \item [Case 2:] the vertex~$x$ is in~$V \setminus W$, in the component of color~$i$ and~$w$ is the last entry in the adjacency list of vertex~$v_i$. Then we delete~$x$ from the adjacency list of vertex~$w$.
 \item [Case 3:] the vertex~$x$ is in~$V \setminus W$, in the component of color~$i$ and~$w$ is not the last entry in the adjacency list of vertex~$v_i$. Then we add~$w$ to the end of adjacency list of~$v_i$ and replace~$x$ by~$v_i$ in the adjacency list of vertex~$w$.
\end{itemize}

Note that if $w$ appears on the list of $v_i$, it must be at the end, since we only process another vertex in $W$ after we finish processing $w$.
Hence, this procedure indeed computes~$G'$ according to 
the definition of~$\wtorso(G, W)$. 
This proves the running time bound of computing the graph~$G'$.

Using the colors of the vertices in~$G$ we can create an array to compute~$\phi$ efficiently. To compute~$\phi^{-1}$ we use an array indexed by the colors which at index~$i$ contains a list of all vertices with color~$i$ in~$G$. Filling the array can be done in linear time.
\end{proof}
To prove \autoref{thm:treewidth-reduction-gen} we reuse some results by
\citet{MOR13}. To state them we need to define the excess of a cut.
\begin{definition}
 Let~$G$ be a graph and~$s, t \in V(G)$. Let~$\ell$ be the minimum size of an
\cvsep{s}{t} in~$G$. The \emph{excess} of an \cvsep{s}{t}~$S$ is~$|S| - \ell$.
\end{definition}
\noindent The following three statements have been proved by \citet{MOR13}; they correspond to Lemma 2.11, Lemma 2.8, and Corollary 2.10, respectively.
\begin{lemma}[\citet{MOR13}]\label{lem:separatorset}
 Let~$s,t$ be two vertices of a graph~$G$ and let $\ell$ be the minimum size of
an \cvsep{s}{t}. For some~$e > 0$ let~$C$ be the union of all inclusion-wise
minimal \cvsep{s}{t}s having excess at most~$e$ (that is, having size at most~$k
= \ell + e$). Then, there is an~$f(\ell, e) \cdot (n + m)$-time algorithm
that returns a set~$C' \supseteq C$ disjoint from~$\{s, t\}$ such
that~$\tw(\torso(G, C')) \leq g(\ell, e)$, for some functions~$f$ and~$g$
depending only on~$\ell$ and~$e$.%
\end{lemma}
\begin{lemma}[\citet{MOR13}]\label{lem:torso-union-tw}
 Let~$G = (V, E)$ be a graph and~$W_1, \ldots, W_r \subseteq V$ and let~$W :=
\bigcup_{i = 1}^r W_i$. We have~$\tw(\torso(G, W)) \leq 1 + \sum_{i = 1}^r
\tw(\torso(G, W_i))$.
\end{lemma}
\begin{corollary}[\citet{MOR13}]\label{cor:torso-add-tw}
 For every graph~$G$ and sets~$C, X \subseteq V(G)$, we have $\tw(\torso(G, C \cup
X)) \leq \tw(\torso(G, C)) + |X|$.
\end{corollary}
Finally we put all the above results together in order to prove \autoref{thm:treewidth-reduction-gen}. We have to only slightly adapt the proof of the treewidth-reduction theorem by~\citet{MOR13}.

\begin{proof}[Proof of \autoref{thm:treewidth-reduction-gen}]
  We first use \autoref{lem:separatorset} to compute a set~$C_{s, t}$ for each 
  pair~$s, t \in T$ such that~$C_{s, t}$ contains each inclusion-wise minimal
  \cvsep{s}{t} of size at most~$k$. Let~$C'$ be the union of all the~$O(|T|^2)$
  sets~$C_{s, t}$. Using the upper bound on the treewidth of
  \autoref{lem:separatorset} and combining it with the bound of
  \autoref{lem:torso-union-tw}, we obtain that~$\torso(G, C')$ has treewidth
  bounded by some function~$g(k, |T|)$ depending only on~$k$ and~$|T|$. Hence, by
  \autoref{cor:torso-add-tw} also~$\torso(G, C' \cup T)$ has treewidth bounded
  by such a function. Using \autoref{pro:wtorso-tw} we have a similar
 bound on the graph~$G^*$ in~$(G^*, \phi) = \wtorso(G, C' \cup T)$. We claim 
that~$(G^*, \phi)$ is a~$(k, T)$-trimmer of~$G$; by the above, the treewidth 
of~$G^*$ is bounded and~$(G^*, \phi)$ can be computed within a time bound as 
claimed by the theorem and it only remains to prove that the two trimmer properties hold. 
For trimmer property~\enuref{enu:trp1}, observe that it follows directly 
from~\autoref{pro:wtorso-sep}. Trimmer property~\enuref{enu:trp2} is also not 
hard to obtain. Consider an inclusion-wise minimal \cvsep{s}{t}~$S$ for~$G$ with $s,t \in T$. 
Clearly, since~$S \subseteq C' \cup T$ and~$\phi$ is one-to-one on~$C' \cup T$ 
we have~$\phi(S) = S$ and by \autoref{pro:wtorso-sep}~$\phi(S)$ is 
a~\cvsep{\phi(s)}{\phi(t)} for~$G^*$. Since there is an \pth{s}{t}~$P$ in~$G - 
S'$ for every~$S' \subsetneq S$ and~$P$ is preserved by contracting edges 
contained in~$V \setminus (C' \cup T)$ (which is disjoint to~$S$), there is also 
a~\pth{\phi(s)}{\phi(t)} in~$G^* - \phi(S')$. Hence, since~$\phi$ is one-to-one 
on~$C' \cup T$ there is a~\pth{\phi(s)}{\phi(t)} in~$G^* - S''$ for any~$S'' 
\subsetneq \phi(S)$. Thus, trimmer property~\enuref{enu:trp2} holds.
\end{proof}

\begin{remark}\label{rem:upper-lower-bound-treewidth}
 The concrete function of~$g'$ such that the treewidth of~$G^*$ is at 
most~$g'(k, |T|)$ depends mainly on the function~$g(\ell, e)$ from 
\autoref{lem:separatorset}. As \citet{MOR13} note, $g = 2^{O(e\ell)}$ and 
they also show that hypercubes yield a lower bound on~$g$ which is exponential 
in~$\ell$ and in~$\sqrt{e}$ (see 
\cite[Remark~2.14]{MOR13}): they note that there are choices of 
vertices~$s$ and~$t$ in an $n$-dimensional hypercube graph such that each 
remaining vertex is contained in an inclusion-wise minimal \cvsep{s}{t} of size 
at most~$n(n-1)$. Hence, taking the torso with respect to the vertex set 
containing $s$, $t$, and their inclusion-wise minimal separators of size at 
most~$n(n-1)$ yields an unchanged hypercube graph. Since hypercube graphs have treewidth 
$\Omega(2^n/\sqrt{n})$~\cite{CK06}, this yields the lower bound on~$g$. The same 
is true for taking the corresponding annotated torso and, hence, there is also a 
corresponding lower bound on~$g'$.
\end{remark}
To prove \autoref{thm2} it now only remains to show \autoref{lem:vbp-weights-tw} which is done in the next section.

\subsubsection{An FPT Algorithm for \VBP{} w.r.t. Treewidth}\label{sec:twalg}

\newcommand{\sep}{\ensuremath{\textsl{Sep}}}

To prove \autoref{lem:vbp-weights-tw} we define a table that can be filled by dynamic programming over a tree-decomposition. Let us fix some more notation.

Let~$G$ be a graph with non-negative integer vertex weights and weight function~$\wf$.
Furthermore, let~$(T, \tau)$ be a tree decomposition for~$G$. For~$t \in V(T)$ denote by~$G_t$ the graph induced by the subtree of~$T$ rooted at~$t$; that is,~$G_t$ is the graph induced by the vertex set~$\bigcup_{t'} \tau(t')$, where the union is taken over all successors~$t'$ of $t$ in~$T$. For a partition~$P$ of some set, a pair of partitions~$P_1, P_2$ of the same set is called a~\emph{splitting} of~$P$ if $P$ is the finest common coarsening of $P_1$ and $P_2$. In other words, the transitive closure of the union of the equivalence relations corresponding to~$P_1$ and~$P_2$ yield the equivalence relation corresponding to~$P$. By~$P - v$ we denote the partition derived from~$P$ by removing $v$~from the part it is contained in.

\newcommand{\pas}{\ensuremath{\textsl{past}}}
\newcommand{\pres}{\ensuremath{\textsl{present}}}
\newcommand{\fut}{\ensuremath{\textsl{future}}}

We define the table~$\sep{}(t, S_t, P_A, P_B, c, \ell)$, where $t \in V(T)$, $S_t \subseteq \tau(t)$, $P_A \cup P_B$ is a partition of~$\tau(t) \setminus S_t$, and~$c, \ell$ are non-negative integers. The entry~$\sep{}(t, S_t, P_A, P_B, c, \ell)$ contains the minimum weight~$\wf(S)$ of an \cvsep{A}{B} $S$ of $G_t$ such that
\begin{compactenum}[i)]
\item $\wf(A) = \ell$,\label{enu:twsize}
\item $S \cap \tau(t) = S_t$,\label{enu:twsep}
\item $S$ is a $(|P_A \cup P_B| + c)$-component separator for~$G_t$,\label{enu:twcomp1}
\item taking the set of intersections of each connected component of~$G_t[A]$ with~$\tau(t)$ yields exactly~$P_A$ (after removing the empty set if present) and analogously for~$G_t[B]$ and~$P_B$.\label{enu:twcomp2}
\end{compactenum}
If no such separator exists, we let~$\sep{}(t, S_t, P_A, P_B, c, \ell) = \bot$. Note that, for non\nobreakdash-$\bot$ values of~$\sep$, the above conditions imply that no two vertices from different parts of~$P_A \cup P_B$ are adjacent in~$G_t$. In the following we will tacitly assume that this is always the case. This is no restriction, since it is easily checkable in $O(\tw^2)$ time.

\begin{lemma}\label{lem:vbp-weights-tw-special}
 Let $G$~be a graph with non-negative integer weights~$\wf$ on the vertices and let $\WF$~be the sum of all weights. Furthermore, let $(T, \tau)$~be a tree decomposition for~$G$ of width~$\tw$ with root~$t$. Then, we can compute all the values~$\sep(t, S_t, P_A, P_B, c, \ell)$ in~$O(\tw^{O(\tw)} \cdot c^2 \cdot \WF^2 \cdot n)$~time.
\end{lemma}
\begin{proof}
  By \cite[Lemma 13.1.3]{Klo94} we may assume that~$(T, \tau)$ is a nice tree decomposition with at most $4n$~bags, that is
  \begin{compactitem}
  \item $T$ is a binary tree,
  \item if a node~$x$ in~$T$ has two children~$y,z$, then $\tau(x) = \tau(y) = \tau(z)$ (in this case~$x$ is
    called a \emph{join node}), and 
  \item if a node~$x$ in~$T$ has one child~$y$, then one of the following situations must hold
    \begin{itemize}
    \item $\tau(x) = \tau(y) \cup \{v\}$ for some~$v \in V(G)\setminus \tau(y)$ (in this case $x$~is called an~\emph{introduce
        node}), or
    \item $\tau(x) \setminus \{v\} = \tau(y)$ for some~$v \in V(G)\setminus \tau(x)$ (in this case $x$~is called a~\emph{forget node}).
    \end{itemize}
  \end{compactitem}

  We prove that for any node~$t$ of~$T$ the corresponding values for~$\sep$ can be computed in~$\tw^{O(\tw)} \cdot c^2 \cdot \WF^2$~time if the values for the children of~$t$ in~$T$ are already known or if~$t$ is a leaf. The result then follows by computing the values~$\sep$ in a bottom-up fashion on the at most~$4n$ nodes of~$T$.

 First if~$t$ is a leaf in~$T$, then~$G_t = G[\tau(t)]$ and it is trivial to 
obtain the values of~$\sep$: Simply check whether~$S_t$ is a separator adhering 
to~\enuref{enu:twsize} through~\enuref{enu:twcomp2} for all sets~$S_t 
\subseteq \tau(t)$, all bipartitions of~$\tau(t) \setminus 
S_t$ into~$W_A$ and~$W_B$, all pairs of 
partitions~$P_A, P_B$ of~$W_A$ and~$W_B$, and all values 
for~$\ell$. The check takes $O(\tw^2)$~time and there are at most $2^\tw$~sets~$S_t$, at most $2^\tw$~bipartitions of~$\tau(t) \setminus 
S_t$ into~$W_A$ and~$W_B$, at most $\tw^{O(\tw)}$~pairs of 
partitions~$P_A, P_B$ of~$W_A$ and~$W_B$, and at most $\WF$~values 
for~$\ell$. Filling 
the entries of~$\sep$ for leaves~$t$ can thus be done in $2^{\tw} \cdot 2^{\tw} 
\cdot \tw^{O(\tw)} \cdot \WF \cdot \tw^2 = \tw^{O(\tw)} \cdot \WF$~time.

  If~$t$ is a join node, consider its children~$t_1, t_2$ in~$T$. We claim that
  \begin{multline}\label{eq:join}
    \sep(t, S_t, P_A, P_B, c, \ell) =\\ \min (\sep(t_1, S_t, P_{A_1}, P_{B_1}, c_1, \ell_1) + \sep(t_2, S_t, P_{A_2}, P_{B_2}, c_2, \ell_2) - \wf(S_t)),
  \end{multline}
  where the minimum is taken over all~$0 \leq \ell_1, \ell_2 \leq \ell$ such that~$\ell_1 + \ell_2 = \ell - \wf(\bigcup P_A)$, over all~$0 \leq c_1, c_2 \leq c$ such that~$c_1 + c_2 = c$, and over all pairs of splittings~$P_{A_1}$,~$P_{A_2}$ and~$P_{B_1}$,~$P_{B_2}$ of the partitions~$P_A$ and~$P_B$, respectively. Let us prove the claim.

  ``$\geq$'': Let~$S$ be an \cvsep{A}{B} for~$G_t$ corresponding to the left-hand side $\sep(t, S_t, P_A, P_B, c, \ell)$ and let~$S_1, S_2$ be the intersections of~$S$ with the vertex sets of~$G_{t_1}, G_{t_2}$, respectively. Let us show that~$S_1, S_2$ both adhere to~\enuref{enu:twsize} through~\enuref{enu:twcomp2} for different but tightly related values of~$P_A$, $P_B$, $c$, and~$\ell$. Clearly, for~$\phi \in \{1, 2\}$ we have that~$S_\phi$ is an \cvsep{A_\phi}{B_\phi} in~$G_{t_\phi}$, where~$A_\phi, B_\phi$ are the intersections of~$A$ and~$B$ with the vertex set of~$G_{t_\phi}$. Furthermore, $S_\phi \cap \tau(t_\phi) = S_t$ that is \enuref{enu:twsep}~holds for~$S_\phi$. It is also clear that the set~$S_\phi$ is an \cvsep{A_\phi}{B_\phi} for~$G_{t_\phi}$ for some pair of partitions~$P_{A_\phi}, P_{B_\phi}$ such that \enuref{enu:twcomp2}~holds. Further, for~$\Gamma \in \{A, B\}$ consider the intersection graph~$F_\Gamma$ of the sets in~$P_{\Gamma_1} \cup P_{\Gamma_2}$. The sets in~$P_\Gamma$ correspond exactly to the connected components of~$F_{\Gamma}$. Hence, the pair~$P_{\Gamma_1}, P_{\Gamma_2}$ is a splitting of~$P_\Gamma$. Next, for~$\phi \in \{1, 2\}$ the set~$S_\phi$ is a $(|P_{A_\phi} \cup P_{B_\phi}| + c_\phi)$-component separator in~$G_{t_\phi}$ for some~$c_\phi$, i.e. \enuref{enu:twcomp1}~holds for this value. Since all connected components left by removing~$S$ from~$G_t$ that do not intersect~$\tau(t)$ are either contained in~$G_{t_1}$ or~$G_{t_2}$, we have~$c_1 + c_2 = c$. Finally, since~$A_1 \cap A_2$ is exactly the vertex set partitioned by each of~$P_{A_1}$, $P_{A_2}$, and~$P_A$, if we set~$\ell_\phi = |A_\phi|$ for~$\phi \in \{1, 2\}$, we have that~\enuref{enu:twsize} holds for $S_\phi$ and~$\ell_\phi$. Furthermore,~$\ell_1 + \ell_2 = \ell - \wf(P_A)$. This proves the ``$\geq$\nobreakdash-inequality part'' of \autoref{eq:join}.

  It is not hard to check that also two separators corresponding to the values on the right hand side of~\autoref{eq:join} give a separator corresponding to the left hand side. Hence~\autoref{eq:join} holds. As to computing~$\sep(t, S_t, P_A, P_B, c, \ell)$, observe that~$c_2$ is fixed once we fix~$c_1$ and analogously for~$\ell_1$ and~$\ell_2$. Thus, $\sep(t, S_t, P_A, P_B, c, \ell)$ can be computed by considering $c \cdot \WF$~times all pairs of splittings of~$P_A$ and~$P_B$. To iterate over all splittings, one can simply iterate over all partitions of the corresponding vertex sets and check whether the partitions induce splittings in $\tw^{O(1)}$~time. Hence the overall running time for computing~$\sep(t, S_t, P_A, P_B, c, \ell)$ is~$c \cdot \WF \cdot \tw^{O(\tw)}$ and computing all values~$\sep$ for a join node~$t$ can be done in $c^2 \cdot \WF^2 \cdot \tw^{O(\tw)}$~time.

  If~$t$ is an introduce node, consider the child~$t'$ of~$t$ and $\{v\} = 
\tau(t) \setminus \tau(t')$. Considering an \cvsep{A}{B} for~$G_t$ corresponding 
to~$\sep(t, S_t, P_A, P_B, c, \ell)$ and the separator it induces in~$G_{t'}$, 
the following is easy to observe.
  \begin{align}\label{eq:introduce}
    \sep(t, S_t, P_A, P_B, c, \ell) &=
    \begin{cases}
      \sep(t', S_t \setminus \{v\}, P_A, P_B, c, \ell) + \wf(v), & \text{ if } v \in S_t \\
      \min \sep(t', S_t, P_A', P_B, c, \ell - \wf(v)), & \text{ if } v \in \bigcup P_A \\
      \min \sep(t', S_t, P_A, P_B', c, \ell), & \text{ if } v \in \bigcup P_B \text{.}
    \end{cases}
  \end{align}
Here the minimum in the second case is taken over all sets $P_A'$ such that $P_A'$ and $\{N[v] \cap (\tau(t) \setminus S_t)\} \cup \{\{u\} \mid u \in \bigcup P_A \setminus N[v]\}$ form a splitting of~$P_A$. Similarly, the minimum in the third case is taken over all sets $P_B'$ such that $P_B'$ and $\{N[v] \cap (\tau(t) \setminus S_t)\} \cup \{\{u\} \mid u \in \bigcup P_B \setminus N[v]\}$ form a splitting of $P_B$.   
Note that in the second case~$v$ has no neighbors in any set of~$P_B$ and in the 
third case~$v$ has no neighbors in any set of~$P_A$. If this is not true then, 
technically \autoref{eq:introduce} may not hold. However, as explained above, we 
may ignore this situation since then, clearly, $\sep(t, S_t, P_A, P_B, c, \ell) 
= \bot$ and it is easily checkable in $O(\tw^2)$ time.
Hence, using \autoref{eq:introduce} we can compute~$\sep(t, S_t, P_A, P_B, c, 
\ell)$ in $c \cdot \WF \cdot \tw^{O(\tw)}$~time for all table entries at an 
introduce node~$t$.

If~$t$ is a forget node, consider its child~$t'$ and $\tau(t') = \tau(t) \cup 
\{v\}$. Consider an \cvsep{A}{B}~$S$ corresponding to $\sep(t, S_t, P_A, P_B, c, 
\ell)$. There are again three cases to consider. First, $v$ can be part of the 
desired minimum separator $S \supseteq S_t$, then~$\sep(t, S_t, P_A, P_B, c, 
\ell) = \sep(t, S_t \cup \{v\}, P_A, P_B, c, \ell)$. Second, $v$ can be in $A$, 
then either~$v$ is not connected to any vertex in~$\tau(t') \setminus S$ 
and~$\sep(t, S_t, P_A, P_B, c, \ell) = \sep(t, S_t, P_A \cup \{v\}, P_B, c - 1, 
\ell)$. Otherwise, $\sep(t, S_t, P_A, P_B, c, \ell) = \min \sep(t, S_t, P'_A, 
P_B, c, \ell)$, where the minimum is taken over all~$P'_A$ derived from~$P_A$ by 
adding~$v$ to a part. The third case~$v \in B$ is analogous to the second case. 
Hence, to compute~$\sep(t, S_t, P_A, P_B, c, \ell)$ we have to simply keep the 
minimum weight assumed in one of the cases.
It is easy to check that, using the above case-distinction, we can also compute all the table entries for forget nodes in~$c \cdot \WF \cdot \tw^{O(\tw)}$~time. This concludes the proof.
  \end{proof}
\autoref{lem:vbp-weights-tw} now follows as an easy corollary.
\begin{proof}[Proof of \autoref{lem:vbp-weights-tw}, Sketch]
  We first compute a nice tree-decomposition~$(T, \tau)$ of width~$O(\tw)$ for the input graph~$G$, which is possible in~$2^{O(\tw)} \cdot n$~time~\cite{BDDFLP13}. We then compute all the values~$\sep(t, S_t, P_A, P_B, c, \ell)$ where~$t$ is the root of~$T$ using \autoref{lem:vbp-weights-tw-special}. Using standard techniques, we retrace the minima in the corresponding dynamic program to find an \cvsep{A}{B}~$S$ and the corresponding sets~$A, B$ for all the values~$\sep(t, S_t, P_A, P_B, c, \ell)$. For every~$1 \leq \ell \leq s$ and for every of the $O(n)$~bags of~$(T, \tau)$ only a constant number of disjoint set union operations is needed to obtain the sets~$A, B$ and~$S$. This amounts to $O(n^3)$~time spent. Hence, deriving the sets does not increase the running time bound of \autoref{lem:vbp-weights-tw-special}.
\end{proof}

\section{Incompressibility of \BP{}}\label{sec:incomp}

 Problem kernelization is a powerful preprocessing tool in attacking NP-hard
problems~\cite{GN07,Bod09}. A \emph{reduction to a problem kernel} is an
algorithm that, given an instance~$I$ with parameter~$p$ of a parameterized
problem, in time polynomial in~$(|I|+p)$ outputs an instance $I'$ of the same
problem and a parameter~$p'$ such that
\begin{compactenum}[i)]
\item  $I$~is a yes-instance if and only if~$I'$ is a yes-instance,
\item $|I'|+p'\leq f(p)$, where $f$~is a function only depending on~$p$.
\end{compactenum}
The function~$f$ is called the \emph{size} of the problem kernel. It is
desirable to find problem kernels of size polynomial in the parameter~$p$.

 In this section, we show that, unless a reasonable complexity-theoretic assumption fails, \BP{} has no polynomial-size kernel with respect
to the cut size (the ``standard parameter'') and any parameter that is polynomial in the input size and does not increase when taking disjoint unions of graphs. Let us call such parameters \emph{union-oblivious}. %
Our result excludes
polynomial-size problem kernels for the parameters
treewidth, 
cliquewidth, or bandwidth, for example. %

\begin{theorem}\label{the:no-poly-kernel-tw-k}
   Unless coNP${}\subseteq{}$NP/poly, \BP{} does not admit polynomial-size kernels with respect to the desired cut size and any union-oblivious parameter. %
\end{theorem}

\noindent To prove \autoref{the:no-poly-kernel-tw-k}, we first show that a
version of \BP{} with integer edge weights does not have a polynomial-size
kernel, and then show how to remove the weights. To obtain that \EWBS{} does
not have a polynomial-size kernel, it is sufficient to show a cross
composition (cf.~\citet{BJK14}) from the NP-hard~\cite{GareyJ79} \MaxCu{}
problem to
\EWBS{}. \MaxCu{} is defined as follows.
\decprob{\MaxCu{}}{A graph~$G = (V, E)$ and an integer~$k$.}{Is there a
partition of~$V$ into sets~$A$ and~$B$ such that at least $k$~edges have one
endpoint in~$A$ and one in~$B$?}
\noindent Showing the cross composition amounts to the following.
We give a polynomial-time algorithm that transforms input instances
$(G_1,k_1),\dots,(G_t,k_t)$ of \MaxCu{} into one instance~$(G^*,k^*)$ of \EWBS{}
such that $(G^*,k^*)$ is a yes-instance if and only if one of the \MaxCu{}
instances is, and such that $k^*$ is polynomial in the size of the largest input instance.

\begin{construction}\label{crossco}
 The construction resembles the reduction given for the NP-hardness of \BP{} by 
\citet{GareyJS76}. To ease the presentation of the construction we assume the following without loss of 
generality.
\begin{compactenum}[i)]
 \item Each of the~$G_i$, $1\leq i\leq t$, has exactly $n$~vertices and 
$k_1=\dots=k_t=:k$. We may assume this because it implies a polynomial-time computable equivalence relation on the instances of \MaxCu{}, see \citet{BJK14}.
 \item It holds that $1\leq k\leq n^2$. Indeed, if $k=0$ then all instances are yes-instances, and if 
$k>n^2$ then all instances are no-instances. Hence, if not $1\leq k\leq n^2$, we can 
return a trivial yes-instance or no-instance of \EWBS{}.
 \item The number~$t$ of input instances is odd. Otherwise, we can add a no-instance to the list of input 
instances that consists of the edgeless graph on $n$ vertices.
\end{compactenum}
We create $G^*$ as follows. For each input graph~$G_i=(V_i,E_i)$,
$1\leq i\leq t$, add to~$G^*$ the vertices in~$V_i$ and a clique $V_i'$ with
$|V_i|$~vertices and edges of weight~$W:=n^2$ each. We make all vertices in~$V_i'$ adjacent to all vertices in~$V_i$ in~$G^*$ via an edge of weight~$W$. Now, for
each pair~$v,w\in V_i$, we add an edge~$\{v,w\}$ to~$G^*$ with weight~$W$ if
$\{v,w\}\notin E_i$ and with weight~$W-1$ if $\{v,w\}\in E_i$. We set
$k^*:=Wn^2-k$.
\end{construction}
\noindent Let us prove that \autoref{crossco} is the promised cross composition.

\begin{lemma}\label{bandwidth-croco}
\autoref{crossco} is a cross composition from \MaxCu{} to \EWBS{} with respect to the desired cut weight and any union-oblivious parameter. %
\end{lemma}
\begin{proof} 
  First, it is clear that the desired cut weight in an instance created by \autoref{crossco} is bounded by a polynomial in~$n$ because~$W = n^2$. Furthermore, let us note that any union-oblivious parameter as above is polynomial in~$n$. This is true because the constructed 
output graph~$G^*$ consists of connected components, each having at most 
$2n$~vertices. (Recall that union-oblivious parameters are polynomial in the input size and do not increase when taking the disjoint union of graphs.)

  It is also clear that \autoref{crossco} can be carried out in polynomial time. It thus remains to show that the instance 
$(G^*,k^*)$ output by \autoref{crossco} is a yes-instance for
\EWBS{} if and only if there is an~$i\in\{1,\dots,t\}$ such that the input instance~$(G_i,k)$ is a
yes-instance for \MaxCu{}.

($\Leftarrow$)
Without loss of generality, let $(G_1,k)$ be a yes-instance for \MaxCu{}. Then,
$V_1=A\uplus B$ such that there are at least $k$~edges with one
vertex in~$A$ and the other in~$B$.\footnote{Here, $\uplus$ denotes the disjoint union of sets.} We show how to construct a solution for
\EWBS{} in~$G^*$ by partitioning~$G^*$ into two vertex sets~$A'$ and~$B'$, where
we choose $A'$ and $B'$ as follows.
 \begin{compactenum}[i)]
 \item $A'$ contains $V_1\cap A$, arbitrary $|B|$~vertices of~$V_1'$, and
$\bigcup_{i=2}^{\lceil t/2\rceil} V_i\cup V_i'$;
 \item $B'$ contains $V_1\cap B$, the $|A|$~vertices of~$V_1'\setminus
A'$, and $\bigcup_{i=\lceil t/2\rceil+1}^{t} V_i\cup V_i'$.
 \end{compactenum}
Obviously, $|A'|=|B'|$ since $t$~is odd. %
We analyze the total weight of edges cut by the partition into~$A'$ and~$B'$. Only
edges between vertices in~$V_1\uplus V_1'$ are cut. The graph induced by
$V_1\uplus V_1'$ is a clique and each of $A'$ and $B'$ contains exactly $n$
vertices of this clique. Since for~$k$ out of the~$n^2$ cut edges we pay the
cheaper weight of~$W-1$ instead of~$W$, the total weight of the cut edges is
$Wn^2-k=k^*$. It follows that $(G^*,k^*)$~is a yes-instance for \EWBS{}.

($\Rightarrow$)
Assume that, for all~$i\in\{1,\dots,t\}$, $(G_i,k)$~is a no-instance for
\MaxCu{}. We analyze the number of edges cut by a bisection of~$G^*$
into~$A\uplus B$. Consider a connected component induced by $V_i\uplus V_i'$ of
$G^*$. Let $a_i=|(V_i\uplus V_i')\cap A|$ be the number of vertices cut from
this component by the bisection. Since $(G_i,k)$~is a no-instance, the maximum
cut in each~$G_i$ cuts at most~$k-1$ edges. Hence the weight of the edges
between $(V_i\uplus V_i')\cap A$ and $(V_i\uplus V_i')\cap B$ is at
least~$Wa_i(2n-a_i)-(k-1)$: for at most $k-1$~edges we pay the cheaper cost
of~$W-1$ instead of~$W$. 

Since $t$~is odd, at least one component of $G^*$ has to be cut.
First, assume that only one of the components of $G^*$ is cut
by the bisection. In this case $a_i=n$ for the corresponding value $i$ since the 
partition into~$A$ and~$B$ has to be balanced, and the cut size of the bisection 
is at least $Wn^2-k+1>k^*$. Hence
$(G^*,k^*)$ is a no-instance.

We now show that there is indeed only one component of $G^*$ that is cut. Assume 
that this is not the case. Then, there are non-zero values $a_i,a_j$ for some 
$i\neq j$. We may assume without loss of generality that $a_i+a_j\leq 2n$. 
Otherwise, we can define the $a_i$ values as the number of vertices cut out by 
$B$ instead of $A$. By the argument above, the total weight of edges cut in both 
components is at least 
$$Wa_i(2n-a_i)+Wa_j(2n-a_j)-2(k-1)=2nW(a_i+a_j)-W(a_i^2+a_j^2)-2(k-1).$$ Now, 
consider cutting all $a_i+a_j$ vertices from only one component instead, which 
is possible since $a_i+a_j\leq 2n$. The total weight of edges cut in the considered 
components would be at most 
$$W(a_i+a_j)(2n-(a_i+a_j))=2nW(a_i+a_j)-W(a_i^2+a_j^2)-2Wa_ia_j.$$ We assumed 
however that $k\leq n^2=W$, and $a_i,a_j\neq 0$. Hence the cut size would drop 
after changing the bisection, which contradicts its minimality.
\end{proof}
\noindent We are now ready to prove \autoref{the:no-poly-kernel-tw-k}.
\begin{proof}[Proof of \autoref{the:no-poly-kernel-tw-k}] It remains to show 
that the instance $(G^*,k^*)$ of \EWBS{} 
resulting from \autoref{crossco} can be converted to a \BP{} instance such that 
the considered parameters remain polynomial in~$n$.

  Let $(G^*,k^*)$~be the instance of \EWBS{} resulting from \autoref{crossco}.
We create an equivalent instance $(G',k^*)$ of \BP{} as follows.
\begin{compactenum}[i)]
  \item For each vertex~$v$ of~$G^*$, introduce a clique~$C_v$ with 
$W+k^*+2$~vertices
to~$G'$.
  \item For an edge~$\{v,w\}$ of~$G^*$ with weight $\omega$, add $\omega$ 
pairwise
disjoint edges from~$C_v$ to~$C_w$.
  \end{compactenum}
Now it is easy to see that $(G^*,k^*)$~is a yes-instance if and only if
$(G',k^*)$~is, since no bisection with cut size at most~$k^*$ can cut a 
clique~$C_v$
that was introduced for a vertex~$v$.

It is clear that the desired cut size is polynomial in~$n$. It remains to show $\phi(G')\in\poly(n)$ for any union-oblivious parameter~$\phi$. To this end, observe that $G'$ is the disjoint union of its connected 
components. Hence, $\phi(G')\leq\phi(C')$ for some connected component~$C'$ 
of~$G'$. By construction of~$G'$ from~$G^*$, have $|C'|\in\poly(n)$, since each 
connected component of~$G^*$ has $\poly(n)$~vertices. Hence, 
$\phi(G')\leq\phi(C')\in\poly(|C'|)\subseteq \poly(n)$.
\end{proof}

\section{\BP{} and the Cliquewidth-\boldmath$q$ Vertex Deletion Number}\label{dcc}

In this section we show that \BP{} is fixed-parameter tractable with respect to
the number of vertices that have to be removed from a graph to reduce its cliquewidth to some
constant~$q$. Thus, we generalize
many well-studied graph parameters like vertex cover number ($q=1$)~\cite{ChenKX10},
cluster vertex deletion number and cograph vertex deletion number
($q=2$)~\cite{CO00}, or feedback vertex set number~($q=3$)~\cite{KLL02} and treewidth-$t$ vertex deletion number~\cite{FLMS12}. Let us formally define cliquewidth. The definition is inspired by
\mbox{\citet{HOSG08}.}
\newcommand{\intro}{\bullet}
\newcommand{\join}{\eta}
\newcommand{\recol}{\rho}
\newcommand{\union}{\oplus}

Let $q$ be a positive integer. We call $(G,\lambda)$ a \emph{$q$-labeled graph}
if $G$ is a graph and $\lambda: V(G) \to\{1,2,\ldots,q\}$ is a mapping. The
number $\lambda(v)$ is called \emph{label} of a vertex $v$. We introduce the
following operations on labeled graphs.
\begin{compactenum}[i)]
\item For every $i$ in $\{1,\ldots,q\}$, we let $\intro_i$ denote the graph
with only one vertex that is labeled by $i$ (a constant operation).
\item For every pair of distinct $i,j\in\{1,2,\ldots,q\}$, we define a
unary operator~$\join_{i,j}$ such that $\join_{i,j}(G, \lambda) = (G',
\lambda)$, where $V(G') = V(G)$, and $E(G') = E(G) \cup \{(v,w) \mid v,w \in V,
\lambda(v) = i, \lambda(w) = j\}$. In other words, the operator adds all edges
between label-$i$ vertices and label-$j$ vertices.
\item For every pair of distinct $i,j\in\{1,2,\ldots,q\}$, we let
$\recol_{i \to j}$ be the unary operator such that $\recol_{i \to j}(G, \lambda)
= (G, \lambda')$, where $\lambda'(v) = j$ if $\lambda(v) = i$, and
$\lambda'(v)=\lambda(v)$ otherwise. The operator only changes the labels of vertices labeled~$i$ to~$j$.
\item Finally, $\union$ is a binary operation that makes the disjoint
union, while keeping the labels of the vertices  unchanged. Note explicitly that
the union is disjoint in the sense that $(G, \lambda) \union (G, \lambda)$ has
twice the number of vertices of $G$.
\end{compactenum}
A \emph{$q$-expression} is a well-formed expression $\varphi$ written with these
symbols. The $q$-labeled graph produced by performing these operations
therefore has a vertex for each occurrence of the constant symbol in~$\varphi$;
and this $q$-labeled graph (and any $q$-labeled graph isomorphic to it) is
called the \emph{value} $val(\varphi)$ of $\varphi$. If a $q$-expression~$\varphi$ has value~$(G, \lambda)$, we say that $\varphi$ is a
\emph{$q$-expression of}~$G$. The \emph{cliquewidth} of a graph~$G$, denoted by
$cwd(G)$, is the minimum $q$ such that there is a $q$-expression of $G$. We say that a join $\join_{i,j}$ is \emph{full} if there is no edge between
vertices of label $i$ and $j$ in the labeled graph on which the join is applied.

\begin{proposition}\label{obs:expr_size}\label{lem:expres_full}
For any $q$-expression for an $n$-vertex graph there is an equivalent one which is at most as long as $\varphi$, contains $O(q^2n)$~symbols, and for which every join is full.
\end{proposition}
\begin{proof}[Proof sketch]
For the first statement, observe that there are $n$ symbols~$\intro_i$ and $n-1$
unions. Between any two
unions, obviously the number of $\recol$'s and $\join$'s can be reduced to
$O(q^2)$.

The second statement follows from inspecting the proof of Corollary 2.17 by \citet{CO00}, which states that for every $q$-expression, there is an equivalent ``irredundant'' one, meaning that every join is full.
\end{proof}

\newcommand{\cut}{{\textsl{Cut}_{A_0, B_0}}}%
\newcommand{\stdcut}{\cut(\varphi, \vec{a}, \vec{b})}
In the following, we show how to compute an optimal bisection using the $q$-expression of a given graph~$G$. This will naturally also solve the decision problem \BP{}.
Let~$D \subseteq V(G)$ and~$\varphi$ be a $q$-expression for~$G - D$, i.e.\ $val(\varphi)= (G - D, \lambda)$. Let~$A_0, B_0$
be a partition of~$D$. For now, we assume that there are no edges between~$A_0$ and~$B_0$. Let~$n_i(\varphi)$ for~$i \in \{1, \ldots, q\}$ be the number of vertices of~$G
- D$ with label~$i$.
For every pair of vectors~$\vec{a} = (a_1, \ldots, a_q)$, $\vec{b} = (b_1, \ldots, b_q) \in \mathbb{N}^q$ 
with~$a_i+b_i = n_i(\varphi)$, let us denote by~$\stdcut$ the minimum number of edges between different parts of a
partition~$(A,B)$ of~$V(G)$ which satisfies the following conditions.
\begin{compactenum}[i)]
\item $A_0 \subseteq A$, $B_0 \subseteq B$, and\label{enu:cwresp}
\item the number of vertices in~$A \setminus D$ and~$B \setminus D$ of label~$i$ are~$a_i$ and~$b_i$, respectively.\label{enu:cwlab}
\end{compactenum}
In the following we use~$x_i$ to denote the~$i$'th entry of a vector~$\vec{x}$.

\begin{lemma}\label{lem:XP_wrt_cw}
For given~$G$, $A_0$, $B_0$ and~$\varphi$ in time~$O(n^{2q}\cdot q \cdot |\varphi|)$ we can compute all the numbers~$\stdcut$.
\end{lemma}%

\begin{proof}
We prove the lemma by induction on the length of the $q$-expression. By
\autoref{lem:expres_full} we can assume that every join in $\varphi$ is full. 
If~$\varphi=\intro_i$, then we have~$n_i(\varphi)=1$ and~$n_j(\varphi)=0$ for
every~$j \neq i$.
Hence, in each pair of $q$-dimensional vectors~$\vec{a}, \vec{b}$ of~$\stdcut$ there is either~$a_i=1$ or~$b_i=1$ and the
other numbers are zero. In this case, there is exactly one partition fulfilling the
conditions~\enuref{enu:cwresp} and~\enuref{enu:cwlab}, namely the one which puts the only vertex of~$G - D$ to set~$A$ or~$B$ as required. It is easy to compute the number of
edges between the parts in this partition.
 
Now, suppose~$\varphi = \join_{i,j}(\varphi')$. Since~$\varphi'$ is shorter than~$\varphi$, by the induction hypothesis we can %
compute all the numbers~$\cut(\varphi', \vec{a}, \vec{b})$ and store them in a table. Note that~$val(\varphi')$ differs from~$G - D$ only in that~$G - D$ has an edge between every
vertex of label~$i$ and every vertex of label~$j$, while~$val(\varphi')$ has no such edges (as the join is full). Therefore, every partition~$(A,B)$ of~$G - D$ fulfilling the conditions~\enuref{enu:cwresp} and~\enuref{enu:cwlab}, is also a partition for~$val(\varphi')$ fulfilling these conditions, but in~$G - D$ there are
exactly~$a_i \cdot b_j+ a_j \cdot b_i$ more edges between the parts. Hence,
we can output $\stdcut =
\cut(\varphi', \vec{a}, \vec{b})+a_i \cdot b_j+ a_j \cdot
b_i$.
 
 Next, let us assume that~$\varphi = \recol_{i \to j}(\varphi')$, and the values of~$\cut(\varphi', \vec{a}', \vec{b}')$ are already computed and stored in a table. Note that in~$G - D$ there are no vertices
of label~$i$, so we have~$0=n_i(\varphi)=a_i=b_i$. On the other hand, some of the
vertices which have label~$j$ in~$G - D$ had label~$i$ in~$val(\varphi')$. 
A minimal partition for~$G - D$, $\vec{a}$, and~$\vec{b}$
which satisfies the conditions~\enuref{enu:cwresp} and~\enuref{enu:cwlab} is also a partition for~$val(\varphi')$ which satisfies the conditions~\enuref{enu:cwresp} and~\enuref{enu:cwlab} for some~$\vec{a}', \vec{b}'$ and a corresponding distribution of~$a_j$
to~$a'_j$ and~$a'_i$ and of~$b_j$ to~$b'_j$ and~$b'_i$. Therefore~$\stdcut$ can be computed as~$\min \{ \cut(\varphi',
\vec{a}', \vec{b}')\}$, where the minimum is taken over all
pairs~$\vec{a}', \vec{b}'$ with~$a'_t = a_t$ and~$b'_t = b_t$ for
every~$t \in \{1, \ldots, q\} \setminus \{i,j\}$, $a_j = a'_{j} + a'_{i}$, $b_j
= b'_{j} + b'_{i}$, and~$a'_t + b'_t =n_t(\varphi')$ for~$t \in \{i,j\}$. As
every pair~$\vec{a}', \vec{b}'$ gives rise to exactly one~$\vec{a}, \vec{b}$, all the minima can be computed in one pass over
all~$\vec{a}', \vec{b}'$.

Finally, let~$\varphi = \varphi^1\union\varphi^2$ and let the values of~$\cut(\varphi^1, \vec{a}^1, \vec{b}^1)$ and $\cut(\varphi^2, \vec{a}^2, \vec{b}^2)$ be already
computed and stored in a table. A minimal partition for~$G - D$ and~$\vec{a}, \vec{b}$ satisfying the conditions~\enuref{enu:cwresp} and~\enuref{enu:cwlab}
also induces partitions for~$val(\varphi^1)$ and~$val(\varphi^2)$, which satisfy
the conditions~\enuref{enu:cwresp} and~\enuref{enu:cwlab} for some~$\vec{a}^1, \vec{b}^1$ and~$\vec{a}^2, \vec{b}^2$ and corresponding distributions of~$a_i$ to~$a^1_i$ and~$a^2_i$ and of~$b_i$ to~$b^1_i$ and~$b^2_i$. Moreover, there are no edges between~$val(\varphi^1)$
and~$val(\varphi^2)$.
 Thus $$\stdcut = \min \{\cut(\varphi^1, \vec{a}^1, \vec{b}^1)\ +\ \cut(\varphi^2, \vec{a}^2, \vec{b}^2)\}\text{,}$$ where the minimum is taken over all~$\vec{a}^1, \vec{b}^1$
and~$\vec{a}^2, \vec{b}^2$ where for every~$i \in \{1,
\ldots, q\}$, $a_i = a^1_i + a^2_i$, $b_i = b^1_i + b^2_i$, and $a^l_i+b^l_i
=n_i(\varphi^l)$ for $l\in\{1,2\}$.
As every pair of pairs~$\vec{a}^1, \vec{b}^1$ and $\vec{a}^2, \vec{b}^2$ gives rise to exactly one pair~$\vec{a}, \vec{b}$, all the minima can be computed in one pass over all
combinations of $\vec{a}^1, \vec{b}^1$ and $\vec{a}^2, \vec{b}^2$.

Concerning the running time, we again argue by induction to show that the
overall time is~$O(n^{2q}\cdot q \cdot |\varphi|)$. If~$\varphi=\intro_i$, then $|\varphi|=1$ and the computation of $\cut$ for the only two
possible pairs of $q$-dimensional vectors takes $O(m+n)\subseteq O(n^{2q}\cdot q)$~time. This
constitutes the induction basis. Otherwise, for any sub-expression $\varphi'$ of
a given expression~$\varphi$, the computation of the table for~$\varphi'$ takes
$O(n^{2q}\cdot q \cdot |\varphi'|)$~time by the induction hypothesis. Observe
that there are~$O(n^q)$ different pairs of $q$-dimensional vectors~$\vec{a}, \vec{b}$ with~$a_i+b_i = n_i(\varphi)$. If~$\varphi = \join_{i,j}(\varphi')$, then
the computation for each pair of vectors takes $O(q)$~time. For~$\varphi = \recol_{i
\to j}(\varphi')$, one pass through the table of~$\varphi'$ is obviously
accomplished in $O(n^{q})$~time, spending $O(1)$~time per entry. Since in both
cases~$|\varphi|=|\varphi'|+1$, this proves the time bound for~$\varphi$ for
these expressions. Finally, if~$\varphi = \varphi^1 \union \varphi^2$ then the
tables for~$\varphi^1$ and~$\varphi^2$ can be computed in $O(n^{2q}\cdot q \cdot
(|\varphi^1|+|\varphi^2|))$~time. Then we cycle over the entries of both tables
and for each combination we spend $O(q)$~time, so this can be accomplished in
$O(n^{2q}\cdot q)$~time. Since~$|\varphi|=|\varphi^1|+|\varphi^2|+1$, also in
this case the algorithm runs in $O(n^{2q}\cdot q \cdot |\varphi|)$~time.
\end{proof}
\noindent Note that \autoref{lem:XP_wrt_cw} yields an XP-algorithm for \BP{} with respect to~$q$ by simply setting~$A_0 = B_0 = \emptyset$. Furthermore, we can derive the following.
\begin{theorem}\label{cor:constant-cliquewidth-deletion-set}
Let~$G$ be a graph, $D \subseteq V(G)$ a vertex subset, and $\varphi$ a $q$-expression for $G
- D$. There is an $O(2^{|D|}\cdot n^{2q+1}q^3)$~time algorithm which
computes the optimal bisection of~$G$.
\end{theorem}

\begin{proof}
It is enough to find the minimum of $\cut(\varphi,\vec{a}, \vec{b})$ over all partitions~$A_0,B_0$ and pairs of $q$-dimensional vectors~$\vec{a}, \vec{b}$ with $|A_0|+\sum_{i=1}^q a_i$ equal to $|B_0|+\sum_{i=1}^q b_i$. 
Since \autoref{lem:XP_wrt_cw} only applies when there are no edges between $A_0$
and $B_0$, we delete them and add the number of them to the sum.
As the size of $\varphi$ is $O(q^2 \cdot n)$ by \autoref{obs:expr_size}, the
running time follows from \autoref{lem:XP_wrt_cw}.
\end{proof}%
 Given~$D$, a~$(2^{2 + 3q} - 1)$-expression for~$G - D$ can be computed in $O(n^4)$~time, where~$q$ is the cliquewidth of~$G - D$~\cite{Oum08}. Thus, \BP{} is 
FPT with respect to the size of any constant-cliquewidth vertex-deletion set 
that is obtainable in FPT time.

\begin{corollary}\label{oft-fpt}
  \BP{} is fixed-parameter tractable with respect to the size of a feedback vertex set, the size of a cluster vertex deletion set, and the size of a treewidth-$t$ vertex deletion set.
\end{corollary}

\section{The Hardness of \BB{}}\label{appendix:W-hard}

In this section we consider the \BB{} problem, for which the vertices of a graph 
need to be partitioned into $d$ parts of equal size. As before, the cut size, 
i.e.\ the number of edges connecting vertices of different parts, needs to be 
minimized. The formal definition is stated in \autoref{sec:intro}. It is easy to 
generalize \autoref{thm2} and \autoref{cor:constant-cliquewidth-deletion-set} to 
\BB{} as follows. At the heart of each of these algorithms is a dynamic program 
which recurses on the structure of the given graph. It fills a table with an 
entry for each subgraph on which the algorithm recurses, and all integers~$a$ 
and~$b$ such that~$a+b$ is the size of the subgraph. Every entry contains the 
optimal way to partition the vertices of a subgraph into two parts of sizes~$a$ 
and~$b$. By expanding the table to store the best way to partition the vertices 
of a subgraph into~$d$ parts of sizes~$a_1,\ldots,a_d$, where now~$\sum_{i=1}^d 
a_i$ is the size of the subgraph, the optimum solution to \BB{} can be found. 
This results in algorithms with additional running time factors in the order of 
$n^{O(d)}$. In particular, the running time achieved by adapting 
\autoref{cor:constant-cliquewidth-deletion-set} to \BB{} is~$O(d^{|D|+1}\cdot 
n^{2(d-1)q+1}q^3 )$. Therefore it is a natural question to ask whether 
corresponding FPT algorithms can be found. We note however, that even for 
forests (that is, for graphs of cliquewidth at most~3~\cite{CO00}) any 
algorithm optimally solving \BB{} has to include a running time factor 
of~$n^{f(d)}$ unless FPT$\phantom{}=\phantom{}$W[1].
\begin{theorem}\label{thm:bb-whard-forests}
The \BB{} problem is W[1]-hard with respect to the number~$d$ of parts in the 
partition, even on forests with maximum degree two.
\end{theorem}
We give a reduction from \UBP{} which is defined as follows.
\decprob{\UBP{}}{Positive integers~$w_1, \ldots, w_\ell, b, C$ each 
encoded in unary.}{Is there an assignment of $\ell$~items with weights $w_1, 
\ldots, w_\ell$ to at most~$b$ bins such that none of the bins exceeds 
weight~$C$?}
\citet{JKMS13} showed that \UBP{} is W[1]-hard with respect to the number~$b$ of bins.
\begin{proof}[Proof of \autoref{thm:bb-whard-forests}]
Let us construct an instance of \BB{} from a \UBP{} instance~$(w_1, \ldots, 
w_\ell, b, C)$. It is clear that~$W := \sum_{i=1}^\ell w_i \leq b \cdot C$ since 
otherwise we may output a trivial no-instance. We may furthermore assume that~$W 
= b \cdot C$ since otherwise we may add $b \cdot C - W$~items of weight~1 each. 
Now the task given by the \UBP{} instance is to find a partition of the items 
into $b$~sets such that each set has weight at most~$C = W/b = \lceil W/b 
\rceil$. Hence, an equivalent instance~$(G, k, d)$ of \BB{} is created by 
taking~$G$ to be the disjoint union of~$\ell$ paths with $w_1, \ldots, 
w_\ell$~vertices respectively, setting~$d = b$ and~$k = 0$.
\end{proof}

As mentioned above, also \autoref{thm2} can be generalized to \BB{}, yielding a running time of $h(c, k) \cdot n^{O(c)}$ to find a balanced $d$\nobreakdash-partition with cut size at most~$k$ that cuts into $c$~connected components. (Note that~$c \geq d$.) We already showed that for \VBP{} an $h(c, k) \cdot n^{O(1)}$\nobreakdash-time algorithm is out of reach (\autoref{vbswhard}). We next show that this is also true for \BB{}.

\begin{theorem}\label{prop:bb-whard-d-treewidth-k}
\BB{} is W[1]-hard with respect to the combined parameter~$(k, c)$, where $k$~is 
the desired cut size and~$c$ is the maximum number of components after removing any set of at most~$k$ 
edges from the input graph.
\end{theorem}

\newcommand{\ECP}{\textsc{Eq\-ui\-table Con\-nect\-ed Par\-ti\-tion}}

We use a slight modification of the construction used by \citet{EFGKRS09} to 
show hardness for
\ECP{}. This problem is defined as follows.
\pagebreak[3]
\decprob{\ECP{}}{A graph~$G$ and a positive integer~$d$.}{Is there a partition of the vertices of~$G$ into $d$~parts~$C_1, \ldots, C_d$ such that for all~$i, j$, $1 \leq i < j \leq d$, we have~$||C_i| - |C_j|| \leq 1$ and $G[C_i]$ is connected?}

We show that for graphs
constructed by the corresponding hardness reduction with some small tweaks, each 
balanced $d$-partition
with cut size at most~$3d/2$ consists of parts that are connected and, thus, we
obtain a hardness reduction for \BB{}. The difference of the construction
of~\citet{EFGKRS09} and \autoref{cons:bb-whard} below lies in making the part sizes precise and giving the upper bound on the cut size. Therefore, in the correctness proof, we may rely on the correctness of
the reduction to \ECP{} and use it to prove correctness also for \BB{}.

The reduction for hardness of \ECP{} is from the W[1]-hard~\MCC{} problem~\cite{FHRV09}. In \MCC{} one
is given a graph~$G = (V, E)$, where each vertex~$v \in V$ is colored by a
color~$c(v) \in \{1, \ldots, s\}$ and the question is whether~$G$ contains a
clique with $s$~vertices
such that each vertex it contains has a distinct color.

\newcommand{\choice}[2]{\ensuremath{(#1, #2)}\nobreakdash-choice}
\begin{construction}\label{cons:bb-whard}
Let $(G = (V, E), c, s)$ be an instance of \MCC{} and denote~$V_i = \{v \in V \mid c(v) = i\}$. We construct an instance of \BB{}  in
four steps. First, we construct a ``skeleton'' graph and then we successively
replace vertices and edges by more complicated gadgets. In each step, we will
refer to the skeleton graph as the graph with all previous
substitutions. The number of parts we are looking for in the instance of \BB{} will be set to~$d := 2s(s - 1)$ and the the cut size to~$k := 3s(s - 1)$ yielding an average of three cut edges incident with each part.

{\bf Step 1:} The skeleton graph contains~$s$ cycles, each with~$2(s - 1)$
vertices. For the~$i$'th cycle, let us denote its vertices ``clock-wise'' by 
$N^i_1,
P^i_1, N^i_2, P^i_2, \ldots, N^i_{s}, P^i_{s}$, where we omit~$N^i_i, P^i_i$ in
the sequence. We call these vertices \emph{anchors}. We sometimes need to refer
to the next index in the sequence; for this we define~\[\suc_i(j) :=
\begin{cases} (i+1) \bmod s, &  j = i - 1\\ 1, & j = s\\ j + 1, &
\text{otherwise.}\end{cases} \] The~$i$'th cycle corresponds to the $i$th color
in the \MCC{} and will serve as a ``vertex-chooser'', choosing a vertex with
color~$i$ to be in the clique. For all~$i,j$ with~$1 \leq i < j \leq s$ we
connect the $i$th and the $j$th cycle by the edges~$\{N^i_j, P^j_i\}$
and~$\{N^j_i, P^i_j\}$. These connections will ``transmit'' the choices of the
clique-vertices to the neighboring cycles and ensure that the chosen vertices
are adjacent. This concludes the description of the skeleton graph. We will call
the~$s$ cycles \emph{vertex chooser}s and the inter-cycle edges
\emph{transmitter}s.

  {\bf Step 2:} We now replace some edges in the vertex choosers by a ``choice''
gadget. 
Let $(A, b)$ be a tuple of an integer~$b > 0$ and a set~$A = \{a_1,
\ldots, a_t\}$ of integers such that~$0 \leq a_i < a_{i + 1}\le b$, $1
\leq i \leq t$.  An \emph{\choice{A}{b}} is a path~$v_1, \ldots, v_{t
  + 1}$ with~$t+1$ vertices, each possibly having additional vertices
pending on it (that is, each vertex can have additional degree-one
neighbors). The number of pending vertices on~$v_i$,~$1 \leq i \leq t$
is determined by
\[\begin{cases}
    a_1, & \text{if } i = 1\text{,} \\
    a_i - a_{i - 1} - 1, & \text{if } 2 \leq i \leq t\text{, and} \\
    b - a_t, & \text{if } i = t + 1\text{.}
  \end{cases}\]

The choices will be cut exactly once between vertex~$v_p$ and
vertex~$v_{p+1}$ by any feasible partition and if they are cut in this way,
observe that, except for the first and the last vertex, they contribute~$a_1 +
\sum_{i = 1}^p((a_i - a_{i - 1} - 1) + 1) = a_p$ vertices to one of the parts
and~$b - a_t + \sum_{i = p + 1}^t((a_i - a_{i - 1} - 1) + 1) = b - a_p$ to the
other part.

Using the definition of choice, for every~$1 \leq i \leq s$, we replace each
edge $\{P^i_j, N^i_{\suc_i(j)}\}$ in the skeleton graph by an~\choice{A_i}{\max 
A_i}: we identify $N^i_{\suc_i(j)}$ and the first vertex of the choice (the 
vertex with the lowest index in the path), and we identify~$P^i_j$ and the last
vertex of the choice. For the definition of~$A_i$, first let~$z_0 = 2 |E| + 10$.
Then,~$A_i = \{p \cdot z_0 \mid 1 \leq p \leq |V_i|\}$. %

Intuitively, the vertex choosers in the skeleton graph choose vertices as
follows. Assume that also each edge~$\{N^i_j, P^i_j\}$ is replaced by an~\choice{A_i}{\max A_i} (in order to represent adjacency between the chosen vertices,
the actual choices used will be more complicated). Furthermore, assume that each
part of any solution partition is connected and contains exactly one of any of
the anchors (vertices~$N^i_j, P^i_j$). Let us ignore the transmitters and
consider the~$i$th vertex chooser. Then, in order to have equal-size parts\footnote{Parts in a feasible partition for \BB{} may not be of equal size; we will consider this issue more closely below.}  the
cut between the parts of two ``neighboring'' anchors~$N^i_j, P^i_j$ has to be
at the same position for every choice in the vertex chooser, that is, every
choice is cut exactly once between the vertex~$v_p$ and~$v_{p+1}$ for some~$p$.
The position of the cut in the choices of vertex chooser~$i$ corresponds to the
chosen clique-vertex for color~$i$ in the \MCC{} instance~$(G, c, s)$.

{\bf Step 3:} In order to represent the adjacency of two chosen vertices, we
now substitute choices for both the transmitter edges and some further edges in
the vertex choosers. Fix
arbitrary one-to-one mappings~$\phi_i : V_i \to \{1, \ldots, |V_i|\}$ and~$\psi
: E \to \{1, \ldots, |E|\}$. We replace the edge between the anchors~$N^i_j,
P^i_j$ in each vertex chooser~$i$,~$1 \leq i \leq s$, and for each~$j$,~$1 \leq
j \leq s, i \neq j$, by the \choice{A^i_j}{|V_i| \cdot z_0 + |E|} by
identifying~$N^i_j$ and the first vertex of the choice and identifying~$P^i_j$
and the last vertex of the choice. The choice is supposed to choose vertices as
the choice used in Step~2 but shall also choose an edge incident with the chosen
vertex and a vertex with color~$j$. Thus, we define~$A^i_j = \{p \cdot z_0 +
\psi(\{u, v\}) \mid \{u, v\} \in E \wedge \phi_i(v) = p \wedge c(u) = j \}$. To
ensure that vertex choosers for different colors agree on the chosen edge, we 
replace each
transmitter edge~$\{N^i_j, P^j_i\}$ by a~\choice{\{1, \ldots, |E|\}}{|E|} by
identifying~$P^j_i$ with the first vertex of the choice and identifying~$N^i_j$
with the last vertex of the choice.

{\bf Step 4:} We now add further pending vertices to anchors (vertices~$N^i_j,
P^i_j$) in order to ensure that the parts of the desired partition are connected
and have equal size, and to ensure that the choices work as intended. 
 We add~$10\cdot z_0 \cdot (|V| + |E|) - z_0 \cdot |V_i| - |E|$ pending vertices to 
each
of~$N^i_j, P^i_j$ for~$i, j, 1 \leq i, j \leq s, i \neq j$. This concludes the construction of the \BB{} instance~$(G', k = 3s(s - 1), d = 2s(s - 1))$.

\end{construction}
\noindent Let us prove that \autoref{cons:bb-whard} is the promised parameterized reduction.
\begin{proof}[Proof of \autoref{prop:bb-whard-d-treewidth-k}]
First, we derive the total number of vertices in~$G'$. The number of pending vertices added to anchors in Step~4 is
\begin{align*}
& \phantom{={}} \sum_{i = 1}^s 2(s - 1) \cdot (10 \cdot z_0 \cdot (|V| + |E|) - z_0 \cdot |V_i| - |E|)\\
&= 2(s-1) \cdot (s \cdot 10 \cdot z_0 \cdot (|V| + |E|) -  z_0 \cdot |V| - s \cdot |E|)\text{.}
\end{align*}
Note that any \choice{A}{b} contains exactly $b + 2$~vertices. Let us count the number of vertices contained in choices in vertex chooser~$i$, $1 \leq i \leq s$, without the anchors:
\begin{align*}
  (s - 1) \cdot z_0 \cdot |V_i| + (s - 1) \cdot (|V_i| \cdot z_0 + |E|)\text{.}
\end{align*}
The number of vertices in choices of vertex choosers thus totals at
\begin{align*}
  2(s - 1) \cdot z_0 \cdot |V| + s \cdot |E|\text{.}
\end{align*}
Counting the number of vertices in transmitters without the anchors yields~$s(s - 1)\cdot |E|$. Adding all the vertices and the anchors we thus obtain~$2 s(s - 1) \cdot (10 \cdot z_0 \cdot (|V| + |E|) + |E|/2 + 1)$. Without loss of generality, we may assume that~$|E|$ is even because otherwise we may simply add an isolated edge to~$G$. Then, note that the number of vertices in~$G'$ is a multiple of~$2s(s-1)$. Thus, any partition of the vertices into $d = 2s(s-1)$~parts such that each part has size at most~$\lceil |V(G')| / d \rceil$ contains only parts of size exactly~$$n_0 := 10 \cdot z_0 \cdot (|V| + |E|) + |E|/2 + 1\text{.}$$ This implies that the required upper bound on the part sizes of \BB{} and the equal-size requirement of \ECP{} coincide on the instances created by \autoref{cons:bb-whard}.

\citet{EFGKRS09} proved that \autoref{cons:bb-whard} is correct if every part
in the desired partition is connected and the cut size can be arbitrary. We
simply prove that setting the cut size to~$3(s-1)s$ ensures connectivity of
each part. For this, we consider one part and the number of cut edges it
contributes. We prove the following.
\begin{quote}
   {\bf Claim:} Each part in a balanced partition has at least three incident
cut edges and if it has exactly three, then it is connected.
\end{quote}
This claim implies that each of the $2(s - 1)s$~parts is connected in a
yes-instance, since otherwise the cut size would be at least~$3(s-1)s + 1$.
Let us now consider a connected component in one of the parts. If it contains
exactly one vertex, then we call it \emph{small} and it has at least one 
incident
cut edge. If it contains at least two vertices but none of the anchors
(vertices~$N^i_j, P^i_j$), then we call it \emph{medium}. Observe that medium 
connected
components have at least two incident cut edges. Note also that medium connected
components have size at most~$\max\{z_0|V_i| + |E| \mid 1 \leq i \leq s\}$
because this number is the maximum number of vertices in a choice. If the
connected component contains at least one of the anchors, then we call it
\emph{large}. It is not hard to see that large connected components have at 
least three incident cut edges. This is clear if the component contains only 
one anchor. If it contains at least two anchors, then, since each part of the 
partition has size exactly~$n_0$ and each anchor has at 
least~$9 z_0 (|V| + |E|)$ pending vertices, there are at least three pending 
vertices cut off. Note also that every connected
component of size greater than~$\max\{z_0|V_i| + |E| \mid 1 \leq i \leq s\}$ is
large.

For the first part of the claim, assume that there is a part with at most two
incident cut edges. Thus, it can either contain two small connected components
or one medium one. In both cases, the part is smaller than~$n_0$ for every large-enough \MCC{} instance which
contradicts the balancedness of the partition.

For the second part of the claim, assume that there is a part with exactly
three incident cut edges and at least two connected components. Since the part size is exactly~$n_0$ we conclude that it contains at
least one large connected component. This contradicts the fact that this part
has at most three incident cut edges.

This finishes the proof of the above claim and, thus, \autoref{cons:bb-whard}
is a reduction from \MCC{} to \BB{}. It is easy to verify that it is computable
in polynomial time. It is also clear that~$k$ is bounded by some function of~$s$ and also $c$~is because the graph obtained by \autoref{cons:bb-whard} is connected. Hence, \autoref{cons:bb-whard} is also a parameterized reduction with respect to these parameters.
\end{proof}

We conjecture that this hardness result can be extended to planar graphs using a 
similar technique as \citet{EFGKRS09} and
even to two-dimensional grid graphs using a specific kind of planar embedding.

Interestingly, it seems pivotal that the treewidth is~$\Omega(s^2)$ for the above construction. Hence, we can not trivially infer that \BB{} is W[1]-hard with respect to~$k$ and~$d$ on trees, for example. This is left as an interesting open question.

\section{\BB{} and the Vertex Cover Number}\label{appendix:BB-FPT}

In the following we present an FPT algorithm for \BB{} and parameter \vcn{}, 
which is the size of a minimum vertex cover of the graph. Recently 
\citet{GanianO13} gave an FPT algorithm for the combined parameters~$\tau$ 
and~$d$ with a running time of~$2^{2^{O(d+\tau)}}+2^\tau n$. We improve on this 
by removing the dependence on~$d$: our algorithm has a running time 
of~$O(\tau^\tau n^3)$. Some of the ideas of our algorithm are inspired 
by~\citet{DouchaK12}.

\begin{theorem}\label{bb-vc-fpt}
  \BB{} is fixed-parameter tractable with respect to the size of a minimum 
vertex cover of the input graph.
\end{theorem}

\begin{proof}
Let a vertex cover $C$ of size \vcn{} be given. We consider all partitions of~$C$ into $d$~sets. Up to isomorphism, there are at most~$\vcn{}^\vcn{}$ possible
such partitions. This can be seen as follows. If~$d\geq\vcn{}$ we can pick
\vcn{} representatives of the $d$ sets, since the maximal number of sets into
which the \vcn{} vertices can be partitioned is~\vcn{}. Otherwise $d$~is upper
bounded by \vcn{}. In both cases each of the \vcn{} vertices can be put into one
of at most \vcn{} sets, which gives the upper bound on the number of partitions.

For each partition of~$C$ we now check whether some set has more than $\lceil
n/d\rceil$ vertices. If this is the case, the partition is discarded. Otherwise
we need to partition the vertex set~$I=V\setminus C$ not belonging to the vertex
cover. For~$v\in I$ let~$c_v(j)$ denote the number of edges that are cut when
putting~$v$ into set~$j$. Since~$I$ is the complement of a vertex cover, it
induces an independent set. Hence the cost~$c_v(j)$ is solely determined by the
partition of~$C$ which at this point is fixed. Also depending on this partition,
each set~$j$ can still hold at most some $s_j$ vertices until it has reached its
full capacity of~$\lceil n/d \rceil$. 

We need to compute an assignment of the vertices in $I$ to sets such that the
capacities of the sets are not exceeded and the total introduced cost is
minimal. This can be done using a minimum cost maximum matching in an auxiliary
graph as follows. Introduce a vertex $w_v$ for each vertex $v\in I$, and $s_j$
vertices $w_j^1,\ldots, w_j^{s_j}$ for each set $j\in\{1,\ldots d\}$. Now for
each $v\in I$, $j\in\{1,\ldots d\}$, and $l\in\{1,\ldots,s_j\}$, connect vertex
$w_v$ with vertex $w_j^l$ using an edge of cost $c_v(j)$.

 Now a min-cost maximum matching in the auxiliary graph corresponds to an
assignment of each vertex $v$ to a set $j$. The resulting partition of the graph
does not have sets containing more than $\lceil n/d \rceil$ vertices. Moreover
this partition has minimum cut size for the fixed partition of the vertex cover,
since the costs of the edges in the auxiliary graph reflect the incurred cut
edges due to the partition of~$I$.

By going through the above steps and picking the best solution among all
partitions of $C$ that are not discarded, the minimum cut size can be computed.
Note that the auxiliary graph is bipartite with at most $3n$ vertices and
$O(n^2)$ edges. Hence the algorithm runs in time $O(\vcn^\vcn\cdot n^3)$, using 
Dijkstra's algorithm in combination with Fibonacci heaps to solve the matching 
problem~\cite{Fredman:1987}.
\end{proof}

\section{Open Problems}\label{sec:concl}

We presented a $h(k,c)\cdot n^{9+c}$-time algorithm for finding a
$c$-component bisection of size at most $k$. However the function
$h(k,c)$ we gave is doubly exponential in $k$, whereas \citet{CLPPS14}
give a $2^{O(k^3)}\cdot n^3 \log^3n$-time algorithm for \BP{}. Hence
even for constant $c$ the latter algorithm improves over ours. However
it is quite conceivable that better running times should be achievable
when combining the parameters $k$ and $c$.  Since $c$ is a small
constant in most practical applications, it would be of real practical
value to find improved FPT algorithms for the \BP{} problem
parameterized by~$k$ and with constant~$c$.

Concerning the \BB{} problem, 
it was already known that it is considerably harder than \BP{}. 
Even for simple graph classes such as trees or grids, the problem is NP-hard to 
approximate~\cite{AF12_J}. It was therefore asked in~\cite{AF12_J} whether 
practical algorithms beyond the standard deterministic worst-case scenario 
exist. In this article we ruled out FPT algorithms for several parameters. 
However the general question remains only partially answered by our negative 
results. One possible direction for further research is finding fixed parameter 
approximation algorithms~\cite{marx2008parameterized} for the problem.

\paragraph{Acknowledgments.} 
René van Bevern and Manuel Sorge gratefully acknowledge support by the DFG, research project DAPA, NI 369/12. Ondřej Such\'y is also grateful for support by the DFG, research project AREG, NI 369/9. Part of Ondřej Such\'y's work was done while with TU Berlin.

The authors thank Bart M. P. Jansen, Stefan Kratsch, Rolf 
Niedermeier and the anonymous referees 
for helpful suggestions. 

\setlength{\bibsep}{0pt}
\bibliographystyle{abbrvnat}
\bibliography{balance}

\begin{thebibliography}{52}
\providecommand{\natexlab}[1]{#1}
\providecommand{\url}[1]{\texttt{#1}}
\expandafter\ifx\csname urlstyle\endcsname\relax
  \providecommand{\doi}[1]{doi: #1}\else
  \providecommand{\doi}{doi: \begingroup \urlstyle{rm}\Url}\fi

\bibitem[Andreev and {R\"acke}(2006)]{RackeA06}
K.~Andreev and H.~{R\"acke}.
\newblock Balanced graph partitioning.
\newblock \emph{Theory of Computing Systems}, 39\penalty0 (6):\penalty0
  929--939, 2006.

\bibitem[Arbenz(2013)]{Arbenz}
P.~Arbenz.
\newblock Personal communication, 2013.
\newblock ETH Zürich.

\bibitem[Arbenz et~al.(2007)Arbenz, van Lenthe, Mennel, M{\"u}ller, and
  Sala]{ArbenzLMMS07}
P.~Arbenz, G.~van Lenthe, U.~Mennel, R.~M{\"u}ller, and M.~Sala.
\newblock Multi-level $\mu$-finite element analysis for human bone structures.
\newblock In \emph{Proceedings of the 8th International Workshop on Applied
  Parallel Computing (PARA~2006)}, volume 4699 of \emph{LNCS}, pages 240--250.
  Springer, 2007.

\bibitem[Bhatt and Leighton(1984)]{BhattL84}
S.~N. Bhatt and F.~T. Leighton.
\newblock A framework for solving {VLSI} graph layout problems.
\newblock \emph{Journal of Computer and System Sciences}, 28\penalty0
  (2):\penalty0 300--343, 1984.

\bibitem[Bodlaender(2009)]{Bod09}
H.~L. Bodlaender.
\newblock Kernelization: New upper and lower bound techniques.
\newblock In \emph{Proceedings of the 4th International Workshop on
  Parameterized and Exact Computation (IWPEC~2009)}, volume 5917 of
  \emph{LNCS}, pages 17--37. Springer, 2009.

\bibitem[Bodlaender et~al.(2013)Bodlaender, Drange, Dregi, Fomin, Lokshtanov,
  and Pilipczuk]{BDDFLP13}
H.~L. Bodlaender, P.~G. Drange, M.~S. Dregi, F.~V. Fomin, D.~Lokshtanov, and
  M.~Pilipczuk.
\newblock An {$O(c^k n)$} 5-approximation algorithm for treewidth.
\newblock In \emph{Proceedings of the 54th Annual IEEE Symposium on Foundations
  of Computer Science (FOCS~2013)}, pages 499--508. IEEE Computer Society,
  2013.

\bibitem[Bodlaender et~al.(2014)Bodlaender, Jansen, and Kratsch]{BJK14}
H.~L. Bodlaender, B.~M.~P. Jansen, and S.~Kratsch.
\newblock Kernelization lower bounds by cross-composition.
\newblock \emph{SIAM Journal on Discrete Mathematics}, 28\penalty0
  (1):\penalty0 277--305, 2014.

\bibitem[Brandes and Fleischer(2009)]{FleischerB09}
U.~Brandes and D.~Fleischer.
\newblock Vertex bisection is hard, too.
\newblock \emph{Journal of Graph Algorithms and Applications}, 13\penalty0
  (2):\penalty0 119--131, April 2009.

\bibitem[Bui and Peck(1992)]{BuiP92}
T.~N. Bui and A.~Peck.
\newblock Partitioning planar graphs.
\newblock \emph{SIAM Journal on Computing}, 21\penalty0 (2):\penalty0 203--215,
  1992.

\bibitem[Bui et~al.(1987)Bui, Chaudhuri, Leighton, and Sipser]{BuiCLS87}
T.~N. Bui, S.~Chaudhuri, F.~T. Leighton, and M.~Sipser.
\newblock Graph bisection algorithms with good average case behavior.
\newblock \emph{Combinatorica}, 7\penalty0 (2):\penalty0 171--191, 1987.

\bibitem[Chandran and Kavitha(2006)]{CK06}
L.~S. Chandran and T.~Kavitha.
\newblock The treewidth and pathwidth of hypercubes.
\newblock \emph{Discrete Mathematics}, 306\penalty0 (3):\penalty0 359--365,
  2006.

\bibitem[Chen et~al.(2010)Chen, Kanj, and Xia]{ChenKX10}
J.~Chen, I.~A. Kanj, and G.~Xia.
\newblock Improved upper bounds for vertex cover.
\newblock \emph{Theoretical Computer Science}, 411\penalty0 (40-42):\penalty0
  3736--3756, 2010.

\bibitem[Courcelle and Olariu(2000)]{CO00}
B.~Courcelle and S.~Olariu.
\newblock Upper bounds to the clique width of graphs.
\newblock \emph{Discrete Applied Mathematics}, 101\penalty0 (1-3):\penalty0
  77--114, 2000.

\bibitem[Cygan et~al.(2014)Cygan, Lokshtanov, Pilipczuk, Pilipczuk, and
  Saurabh]{CLPPS14}
M.~Cygan, D.~Lokshtanov, M.~Pilipczuk, M.~Pilipczuk, and S.~Saurabh.
\newblock Minimum bisection is fixed parameter tractable.
\newblock In \emph{Proceedings of the 46th Annual Symposium on the Theory of
  Computing (STOC 2014)}, 2014.
\newblock To appear.

\bibitem[Delling et~al.(2011)Delling, Goldberg, Pajor, and
  Werneck]{DellingGPW11}
D.~Delling, A.~V. Goldberg, T.~Pajor, and R.~F.~F. Werneck.
\newblock Customizable route planning.
\newblock In \emph{Proceedings of the 10th International Symposium on
  Experimental Algorithms (SEA~2011)}, volume 6630 of \emph{LNCS}, pages
  376--387. Springer, 2011.

\bibitem[Delling et~al.(2012)Delling, Goldberg, Razenshteyn, and
  Werneck]{DFGRW13}
D.~Delling, A.~V. Goldberg, I.~Razenshteyn, and R.~F.~F. Werneck.
\newblock Exact combinatorial branch-and-bound for graph bisection.
\newblock In \emph{Proceedings of the 14th Workshop on Algorithms Engineering
  and Experiments (ALENEX~2012)}, pages 30--44, 2012.

\bibitem[Diestel(2010)]{Diestel10}
R.~Diestel.
\newblock \emph{Graph Theory}, volume 173 of \emph{Graduate Texts in
  Mathematics}.
\newblock Springer, 4th edition, 2010.

\bibitem[Doucha and Kratochv\'{\i}l(2012)]{DouchaK12}
M.~Doucha and J.~Kratochv\'{\i}l.
\newblock Cluster vertex deletion: A parameterization between vertex cover and
  clique-width.
\newblock In \emph{Proceedings of the 37th International Symposium on
  Mathematical Foundations of Computer Science (MFCS~2012)}, volume 7464 of
  \emph{LNCS}, pages 348--359. Springer, 2012.

\bibitem[Downey and Fellows(2013)]{DF13}
R.~G. Downey and M.~R. Fellows.
\newblock \emph{Fundamentals of Parameterized Complexity}.
\newblock Springer, 2013.

\bibitem[Enciso et~al.(2009)Enciso, Fellows, Guo, Kanj, Rosamond, and
  Such{\'y}]{EFGKRS09}
R.~Enciso, M.~R. Fellows, J.~Guo, I.~A. Kanj, F.~A. Rosamond, and O.~Such{\'y}.
\newblock What makes equitable connected partition easy.
\newblock In \emph{Proceedings of the 4th International Workshop on
  Parameterized and Exact Computation (IWPEC~2009)}, volume 5917 of
  \emph{LNCS}, pages 122--133. Springer, 2009.

\bibitem[Espelage et~al.(2001)Espelage, Gurski, and Wanke]{EGW01}
W.~Espelage, F.~Gurski, and E.~Wanke.
\newblock How to solve {NP}-hard graph problems on clique-width bounded graphs
  in polynomial time.
\newblock In \emph{Proceedings of the 27th International Workshop on
  Graph-Theoretic Concepts in Computer Science (WG~2001)}, volume 2204 of
  \emph{LNCS}, pages 117--128. Springer, 2001.

\bibitem[Feldmann(2013)]{AF12_J}
A.~E. Feldmann.
\newblock Fast balanced partitioning is hard, even on grids and trees.
\newblock \emph{Theoretical Computer Science}, 485:\penalty0 61--68, 2013.

\bibitem[Feldmann and Foschini(2012)]{AFFoschini12}
A.~E. Feldmann and L.~Foschini.
\newblock Balanced partitions of trees and applications.
\newblock In \emph{Proceedings of the 29th International Symposium on
  Theoretical Aspects of Computer Science (STACS~2012)}, volume~14 of
  \emph{LIPIcs}, pages 100--111. Dagstuhl, 2012.

\bibitem[Feldmann and Widmayer(2011)]{AFWidmayer11}
A.~E. Feldmann and P.~Widmayer.
\newblock An {$O(n^4)$} time algorithm to compute the bisection width of solid
  grid graphs.
\newblock In \emph{Proceedings of the 19th Annual European Symposium on
  Algorithms (ESA~2011)}, volume 6942 of \emph{LNCS}, pages 143--154. Springer,
  2011.

\bibitem[Fellows et~al.(2009)Fellows, Hermelin, Rosamond, and Vialette]{FHRV09}
M.~R. Fellows, D.~Hermelin, F.~A. Rosamond, and S.~Vialette.
\newblock On the parameterized complexity of multiple-interval graph problems.
\newblock \emph{Theoretical Computer Science}, 410\penalty0 (1):\penalty0
  53--61, 2009.

\bibitem[Flum and Grohe(2006)]{FG06}
J.~Flum and M.~Grohe.
\newblock \emph{Parameterized Complexity Theory}.
\newblock Springer, 2006.

\bibitem[Fomin et~al.(2010)Fomin, Golovach, Lokshtanov, and Saurabh]{FGLS10}
F.~V. Fomin, P.~A. Golovach, D.~Lokshtanov, and S.~Saurabh.
\newblock Algorithmic lower bounds for problems parameterized with
  clique-width.
\newblock In \emph{Proceedings of the 21st Annual ACM-SIAM Symposium on
  Discrete Algorithms (SODA~2010)}, pages 493--502. SIAM, 2010.

\bibitem[Fomin et~al.(2012)Fomin, Lokshtanov, Misra, and Saurabh]{FLMS12}
F.~V. Fomin, D.~Lokshtanov, N.~Misra, and S.~Saurabh.
\newblock Planar $\mathcal{F}$-deletion: Approximation, kernelization and
  optimal {FPT} algorithms.
\newblock In \emph{Proceedings of the 53rd Annual IEEE Symposium on Foundations
  of Computer Science (FOCS~2012)}, pages 470--479. IEEE Computer Society,
  2012.

\bibitem[Fredman and Tarjan(1987)]{Fredman:1987}
M.~Fredman and R.~Tarjan.
\newblock Fibonacci heaps and their uses in improved network optimization
  algorithms.
\newblock \emph{Journal of the ACM}, 34\penalty0 (3):\penalty0 596--615, 1987.

\bibitem[Ganian and Obdr{\v z}{\'a}lek(2013)]{GanianO13}
R.~Ganian and J.~Obdr{\v z}{\'a}lek.
\newblock Expanding the expressive power of monadic second-order logic on
  restricted graph classes.
\newblock In \emph{Proceedings of the International Workshop on Combinatorial
  Algorithms (IWOCA~2013)}, 2013.
\newblock To appear.

\bibitem[Garey and Johnson(1979)]{GareyJ79}
M.~R. Garey and D.~S. Johnson.
\newblock \emph{Computers and Intractability: A Guide to the Theory of
  NP-Completeness}.
\newblock W. H. Freeman and Co., 1979.

\bibitem[Garey et~al.(1976)Garey, Johnson, and Stockmeyer]{GareyJS76}
M.~R. Garey, D.~S. Johnson, and L.~J. Stockmeyer.
\newblock Some simplified {NP}-complete graph problems.
\newblock \emph{Theoretical Computer Science}, 1\penalty0 (3):\penalty0
  237--267, 1976.

\bibitem[Guo and Niedermeier(2007)]{GN07}
J.~Guo and R.~Niedermeier.
\newblock Invitation to data reduction and problem kernelization.
\newblock \emph{SIGACT News}, 38\penalty0 (1):\penalty0 31--45, 2007.

\bibitem[Hliněný et~al.(2008)Hliněný, Oum, Seese, and Gottlob]{HOSG08}
P.~Hliněný, S.~Oum, D.~Seese, and G.~Gottlob.
\newblock Width parameters beyond tree-width and their applications.
\newblock \emph{The Computer Journal}, 51\penalty0 (3):\penalty0 326--362,
  2008.

\bibitem[Jansen et~al.(2013)Jansen, Kratsch, Marx, and Schlotter]{JKMS13}
K.~Jansen, S.~Kratsch, D.~Marx, and I.~Schlotter.
\newblock Bin packing with fixed number of bins revisited.
\newblock \emph{Journal of Computer and System Sciences}, 79\penalty0
  (1):\penalty0 39--49, 2013.

\bibitem[Karypis and Kumar(1998)]{KarypisK1998}
G.~Karypis and V.~Kumar.
\newblock A parallel algorithm for multilevel graph partitioning and sparse
  matrix ordering.
\newblock \emph{Journal of Parallel and Distributed Computing}, 48\penalty0
  (1):\penalty0 71--95, 1998.

\bibitem[Khot and Vishnoi(2005)]{KhotV05}
S.~A. Khot and N.~K. Vishnoi.
\newblock The {U}nique {G}ames {C}onjecture, integrality gap for cut problems
  and embeddability of negative type metrics into $\ell_1$.
\newblock In \emph{Proceedings of the 46th Annual IEEE Symposium on Foundations
  of Computer Science (FOCS~2005)}, pages 53--62. IEEE Computer Society, 2005.

\bibitem[Kloks(1994)]{Klo94}
T.~Kloks.
\newblock \emph{Treewidth -- Computations and Approximations}, volume 842 of
  \emph{LNCS}.
\newblock Springer, 1994.

\bibitem[Kloks et~al.(2002)Kloks, Lee, and Liu]{KLL02}
T.~Kloks, C.~M. Lee, and J.~Liu.
\newblock New algorithms for $k$-face cover, $k$-feedback vertex set, and
  $k$-disjoint cycles on plane and planar graphs.
\newblock In \emph{Proceedings of the 28th International Workshop on
  Graph-Theoretic Concepts in Computer Science (WG~2002)}, volume 2573 of
  \emph{LNCS}, pages 282--295. Springer, 2002.

\bibitem[Kwatra et~al.(2003)Kwatra, Sch{\"o}dl, Essa, Turk, and
  Bobick]{KwatraSETB03}
V.~Kwatra, A.~Sch{\"o}dl, I.~Essa, G.~Turk, and A.~Bobick.
\newblock Graphcut textures: Image and video synthesis using graph cuts.
\newblock \emph{ACM Transactions on Graphics}, 22\penalty0 (3):\penalty0
  277--286, 2003.

\bibitem[Lipton and Tarjan(1980)]{LiptonT80}
R.~J. Lipton and R.~E. Tarjan.
\newblock Applications of a planar separator theorem.
\newblock \emph{SIAM Journal on Computing}, 9:\penalty0 615--627, 1980.

\bibitem[{MacGregor}(1978)]{MacG78}
R.~M. {MacGregor}.
\newblock \emph{On Partitioning a Graph: a Theoretical and Empirical Study.}
\newblock PhD thesis, University of California, Berkeley, 1978.

\bibitem[Marx(2006)]{Mar06}
D.~Marx.
\newblock Parameterized graph separation problems.
\newblock \emph{Theoretical Computer Science}, 351\penalty0 (3):\penalty0
  394--406, 2006.

\bibitem[Marx(2008)]{marx2008parameterized}
D.~Marx.
\newblock Parameterized complexity and approximation algorithms.
\newblock \emph{The Computer Journal}, 51\penalty0 (1):\penalty0 60--78, 2008.

\bibitem[Marx et~al.(2013)Marx, O'Sullivan, and Razgon]{MOR13}
D.~Marx, B.~O'Sullivan, and I.~Razgon.
\newblock Finding small separators in linear time via treewidth reduction.
\newblock \emph{ACM Transactions on Algorithms}, 9\penalty0 (4):\penalty0 30,
  2013.

\bibitem[Niedermeier(2006)]{Nie06}
R.~Niedermeier.
\newblock \emph{Invitation to Fixed-Parameter Algorithms}.
\newblock Oxford University Press, 2006.

\bibitem[Oum(2008)]{Oum08}
S.~Oum.
\newblock Approximating rank-width and clique-width quickly.
\newblock \emph{ACM Transactions on Algorithms}, 5\penalty0 (1), 2008.

\bibitem[R{\"a}cke(2008)]{Racke08}
H.~R{\"a}cke.
\newblock Optimal hierarchical decompositions for congestion minimization in
  networks.
\newblock In \emph{Proceedings of the 40th Annual ACM Symposium on Theory of
  Computing (STOC~2008)}, pages 255--264. ACM, 2008.

\bibitem[Soumyanath and Deogun(1990)]{SoumyanathD90}
K.~Soumyanath and J.~S. Deogun.
\newblock On the bisection width of partial $k$-trees.
\newblock In \emph{Proceedings of the 20th Southeastern Conference on
  Combinatorics, Graph Theory, and Computing}, volume~74 of \emph{Congressus
  Numerantium}, pages 25--37. Utilitas Mathematica Publishing, 1990.

\bibitem[van Bevern et~al.(2013)van Bevern, Feldmann, Sorge, and
  Such\'y]{BFSS13}
R.~van Bevern, A.~E. Feldmann, M.~Sorge, and O.~Such\'y.
\newblock On the parameterized complexity of computing graph bisections.
\newblock In \emph{Proceedings of the 39th International Workshop on
  Graph-Theoretic Concepts in Computer Science (WG '13)}, volume 8165 of
  \emph{LNCS}, pages 76--88. Springer, 2013.

\bibitem[Werneck(2013)]{Werneck}
R.~F.~F. Werneck.
\newblock Personal communication, 2013.
\newblock Microsoft Research Silicon Valley.

\bibitem[Wiegers(1990)]{Wie90}
M.~Wiegers.
\newblock The $k$-section of treewidth restricted graphs.
\newblock In \emph{Proceedings of the 15th International Symposium on
  Mathematical Foundations of Computer Science (MFCS~1990)}, volume 452 of
  \emph{LNCS}, pages 530--537. Springer, 1990.

\end{thebibliography}

\newpage
\appendix

\end{document}